\documentclass[11pt]{article}

\usepackage[margin=1in]{geometry}
\usepackage{mathpazo}
\usepackage{amssymb}
\usepackage{amsmath}
\usepackage{mathtools}
\usepackage{fancyhdr}
\usepackage{tikz}
\usepackage{amsthm}
\usepackage{hyperref}
\theoremstyle{plain}
\usepackage{comment}
\usepackage{url}

\newtheorem{theorem}{Theorem}
\newtheorem{lemma}[theorem]{Lemma}
\newtheorem{problem}{Open Problem}
\newtheorem{corollary}[theorem]{Corollary}
\newtheorem{proposition}[theorem]{Proposition}
\newtheorem*{theorem*}{Theorem}

\theoremstyle{definition}

\newtheorem{example}{Example}

\theoremstyle{remark}

\newcommand{\real}{\mathbb{R}}

\DeclareMathOperator{\OR}{OR}

\DeclareMathOperator{\adv}{Adv}
\DeclareMathOperator{\opt}{opt}

\DeclarePairedDelimiterX{\inp}[2]{\langle}{\rangle}{#1, #2}

\begin{document}
\author{Hamoon Mousavi\\ \href{mailto:sh2mousa@uwaterloo.ca}{sh2mousa@uwaterloo.ca} }
\title{Lower Bounds on Regular Expression Size}
\maketitle
\begin{abstract}
   We introduce linear programs encoding regular expressions of finite languages. We show that, given a language, the optimum value of the associated linear program is a lower bound on the size of any regular expression of the language. Moreover we show that any regular expression can be turned into a dual feasible solution with an objective value that is equal to the size of the regular expression. For binomial languages we can relax the associated linear program using duality theorem. We use this relaxation to prove lower bounds on the size of regular expressions of binomial and threshold languages.
\end{abstract}

\section{Introduction}\label{section:introduction}

Ellul et al. \cite{ellul-regular-expressions} introduced \emph{binomial languages} $B(n,k)=\{s\in \{0,1\}^n:|s|_1=k\}$ and \emph{threshold languages} $T(n,k)=\{s\in\{0,1\}^n:|s|_1\geq k\}$.\footnote{In that paper $T(n,1)=\{s\in\{0,1\}^n:|s|_1\geq 1\}=\{0,1\}^n\setminus\{0^n\}$ is called the omit language.} They also obtained regular expressions of length $O(n\log^k(n))$ for both these languages. The goal of this note is to show that these regular expressions are asymptotically optimal. 

Given a finite language $L$, in Section \ref{section:formulations} we give a maximization linear programming problem $\mathcal{P}(L)$ and show that the optimum value, denoted by $\opt(\mathcal{P}(L))$, is a lower bound on the length of regular expressions of $L$. This is the formulation we use in the next few sections to prove lower bounds on the length of regular expressions of binomial languages. However as we see, also in Section \ref{section:formulations}, this formulation fails to prove any nontrivial lower bounds on threshold languages. So we give a stronger formulation $\mathcal{P}_S(L)$ that can be used to prove lower bounds in the case of threshold languages. In fact we show that $\opt(\mathcal{P}_S(L))$ is the same as the length of the optimal regular expressions of $L$ for any language $L$. We show that $\mathcal{P}_S(T(n,k))$ and $\mathcal{P}(B(n,k))$ are connected and demonstrate that lower bounds on binomial languages can be turned into lower bounds on threshold languages.

In Section \ref{section:duals} we give dual linear programs of $\mathcal{P}(L)$ and $\mathcal{P}_S(L)$ denoted by $\mathcal{D}(L)$ and $\mathcal{D}_S(L)$ respectively. We show that any regular expression $R$ of $L$ can be turned into a feasible solution of $\mathcal{D}(L)$ (and also $\mathcal{D}_S(L)$). Moreover the objective value of the dual feasible solution obtained in this way is $|R|$. 

From Section \ref{section:formulations}, a strategy for proving lower bounds on $L$ is to guess a feasible solution of $\mathcal{P}(L)$ and calculate its objective value (which is automatically a lower bound on the length of regular expressions of $L$). However this is not an easy task as the number of constraints that must be verified (to establish feasibility) can be very large. For some languages $L\subseteq \Sigma^n$ the number of constraints in $\mathcal{P}(L)$ is exponential in $n$. In Section \ref{section:duality} we obtain a relaxation $\mathcal{P}'(n,k)$ of $\mathcal{P}(B(n,k))$ with only polynomially many constraints and for which $\opt(\mathcal{P}'(n,k)) \geq \opt(\mathcal{P}(B(n,k)))$. We use the duality theorem of linear programming and prove that the optimum value $\opt(\mathcal{P}'(n,k))$ is also a lower bound on the length of regular expressions of $B(n,k)$. The problem $\mathcal{P}'(n,k)$ is efficiently solvable, and on all the examples that we tried the quantity $\opt(\mathcal{P}'(n,k))$ equals the length of regular expression of $B(n,k)$ given by Ellul et al. \cite{ellul-regular-expressions}. Therefore we conjecture that these regular expressions are optimal (and not just optimal up to a constant multiplicative factor). 

The Python program solving $\mathcal{P}'(n,k)$ can be downloaded from GitHub at \\\url{https://github.com/hamousavi/Linear-program-of-binomial-languages.git}. This program uses the CVXPY \cite{cvxpy} library for solving convex optimization problems.

In Section \ref{section:feasible} we construct a feasible solution of $\mathcal{P}'(n,k)$ with an objective value of $\Omega(n\log^k(n))$ for $k \leq 3$. Consequently Ellul et al. regular expressions are asymptotically optimal for $k\leq 3$. We conjecture that the same construction gives feasible solutions of $\mathcal{P}'(n,k)$ with objective values of $\Omega(n\log^k(n))$ for all values of $k$.

In Section \ref{section:complexity} we discuss the computational complexity of solving $\mathcal{P}(L)$ and $\mathcal{P}_S(L)$. In Section \ref{section:future} we discuss some open problems.

Chistikov et al. \cite{chistikov_et_al:LIPIcs:2017:7010} (see Appendix D therein) used linear programming to find lower bounds on the length of regular expressions of the language $L_n = \cup_{0\leq i<j<n} a_i a_j$ over the alphabet $\Sigma_n = \{a_0,a_1,\ldots,a_{n-1}\}$. They introduced an integer programming formulation for the length of the regular expressions of $L_n$. They then showed that the linear programming relaxation of this integer programming can be used to give a precise lower bound of $n(\lfloor \log(n)\rfloor+2)-2^{\lfloor \log(n)\rfloor +1}$ on the length of regular expressions of $L_n$.

Parts of this note, in particular the derivations in Sections \ref{section:formulations} and \ref{section:duals}, are inspired by the work of Reichardt \cite{Reichardt:2009} on quantum query complexity. Given a Boolean function $f$ there exists a maximization semidefinite program $\adv(f)$ such that its optimum value is a lower bound on the quantum query complexity of $f$. This is known as the adversary lower bound. The dual of $\adv(f)$ is related to the powerful framework of span programs. Quite interestingly any dual feasible solution can be turned into a quantum algorithm for computing $f$. Moreover, for the algorithm obtained in this way the query complexity is the same as the the objective value of the dual feasible solution. Reichardt gave as example the primal and dual semidefinite programs associated with the threshold function $f_{n,k}$. These are the Boolean functions defined by the relation $f_{n,k}^{-1}(1) = T_{n,k}$. It was known from Beals et al. \cite{Beals:2001} that the quantum query complexity of $f_{n,k}$ is $\Theta(\sqrt{k(n-k+1)})$, and adversary framework provided an alternative proof of this fact.

\section{Preliminaries}\label{section:preliminaries}

An \emph{alphabet} is a finite and nonempty set. The elements of an alphabet are called \emph{alphabet symbols}. In this note, a \emph{string} is a nonempty and finite sequence of alphabet symbols, and a \emph{language} is a nonempty and finite set of strings. Many of our examples in this note are over  the binary alphabet $\{0,1\}$. We use notation $\Sigma^+$ to denote the set of all strings over an alphabet $\Sigma$. Likewise $\Sigma^n$ denotes the set of all strings of length $n$ and $\Sigma^{\leq n}=\Sigma^1\cup\cdots\cup\Sigma^n$ denotes the set of all strings of length at most $n$. For a string $s$, its length is denoted by $|s|$ and the frequency of an alphabet symbol $a$ is denoted by $|s|_a$. For an alphabet symbol $a$ and a positive integer $n$ we use the shorthand notation $a^n$ to refer to the string $\underbrace{a\cdots a}_{n \text{ times}}$.

\emph{Regular expressions} are defined inductively. We say $R$ is a \emph{regular expression} over an alphabet $\Sigma$ if either $R\in\Sigma,R=(R_1+R_2)$ or $R=R_1R_2$ where $R_1$ and $R_2$ are regular expressions over $\Sigma$. This definition deviates from the conventional definition as it disallows Kleene closure $\ast$, the empty string $\epsilon$, and the empty set $\emptyset$. Moreover the parentheses are more constrained in our definition. Languages and regular expressions are associated as usual, and for a regular expression $R$ the language associated with it is denoted by $L(R)$. Concatenation takes precedence over union when interpreting regular expressions, and with this in mind, some parentheses can be dropped, for example $((00+01)+10)1$ can be written as $(00+01+10)1$. The length (also called size) of a regular expression $R$ denoted by $|R|$ is the number of alphabet symbols appearing in $R$. We use the shorthand notation $R^n$ to refer to the regular expression $\underbrace{R\cdots R}_{n \text{ times}}$. It should be understood that $|R^n|=n|R|$. An \emph{optimal regular expression} for a language is one with minimal length which need not be unique. 

For convenience we sometimes use a language in places where a regular expression is expected. For example for languages $L_1,L_2$ we may write $L_1+L_2$ to mean $L_1\cup L_2$. Similarly we may write $sL$ for a string $s$ and a language $L$ to mean $\{st:t\in L\}$. 

The \emph{subexpressions} of a regular expression $R$ are regular expressions themeselves and are defined recursively as follows. The expression $R$ is a subexpression. If $R=(R_1+R_2)$ or $R=R_1R_2$ then $R_1$ and $R_2$ are subexpressions. Every subexpression of these are also subexpressions of $R$. For example all subexpressions of $(00+11)1$ are $0,1,00,11,(00+11),(00+11)1$.  
A \emph{term} in a regular expression is a maximal subexpression consisting only of alphabet symbols (and the implicit concatenations in between). Maximal here means that the subexpression does not belong to a longer subexpression consisting only of alphabet symbols. A term uniquely defines a string. In our example the only terms are $00,11,1$.

The \emph{regular expression closure}, or \emph{closure} for short, of a language $L$, denoted by $\mathcal{C}(L)$, is a set of languages and is defined inductively as follows. The set $\mathcal{C}(L)$ contains $L$. In addition, for every choice of languages $L_1$ and $L_2$ for which $L=L_1L_2$ or $L=L_1+L_2$, the closure $\mathcal{C}(L)$ contains $L_1$, $L_2$, and recursively the sets $\mathcal{C}(L_1)$ and $\mathcal{C}(L_2)$. For example the closure of the language $L((00+11)1)=\{001,111\}$ is the set of languages $\{\{0\},\{1\},\{00\},\{11\},\{01\},\{001\},\{111\},\{00,11\},\{001,111\}\}$. To save space we use regular expressions in place of languages when we write closures, i.e., for our example we write the closure instead as $\{0,1,00,11,01,001,111,00+11,001+111\}$. For languages $K$ and $L$ for which $K\subseteq L$, it evidently holds that $\mathcal{C}(K)\subseteq \mathcal{C}(L)$.

\begin{example}\label{ex:closure all 1}
    Let $n$ be a positive integer. It is easy to verify that $$\mathcal{C}(\Sigma^n)=\{\emptyset\subset L\subseteq\Sigma^{m}:0< m\leq n\}.$$ This is the set of all languages for which all strings are of the same length and that length is at most $n$. 
\end{example}

\begin{example}[binomial languages]\label{ex:closure binomial 1} Let $\Sigma=\{0,1\}$ be the binary alphabet, let $n>0$, let $0\leq k\leq n$, and define a language $B(n,k) = \{s\in\Sigma^n:|s|_1 = k\}$. These are called \emph{binomial languages}. When we write $B(n,k)$ it is implicitly assumed that $n,k$ satisfy the constraints $n>0$ and $0\leq k\leq n$. From our earlier remark it holds that $\mathcal{C}(B(n,k))\subseteq \mathcal{C}(\Sigma^n)$. We can however specify the closure of binomial languages more precisely. In fact it is easy to verify that
	\begin{align*}
	\mathcal{C}(B(n,k)) = \{\emptyset\subset L\subseteq B(m,l): 0< m\leq n, 0\leq l\leq \min(m,k)\}.
	\end{align*}
When working with binomial language the set of pairs $\{(m,l): 0< m\leq n, 0\leq l\leq \min(m,k)\}$ appears a lot, so we introduce the shorthand notation $$\mathcal{C}(n,k)=\{(m,l): 0< m\leq n, 0\leq l\leq \min(m,k)\}.$$
Again when we write $\mathcal{C}(n,k)$ it is implicitly assumed that $n>0$ and $0\leq k\leq n$. 
\end{example}


We often refer to pairs of languages $K_1,K_2\in \mathcal{C}(L)$ for which $K_1K_2\in \mathcal{C}(L)$ as well. So we introduce the shorthand notation $$\mathcal{C}_c(L)=\{(K_1,K_2):K_1,K_2,K_1K_2\in\mathcal{C}(L)\},$$
where the letter $c$ in the subscript refers to concatenation. Likewise we introduce the shorthand notation
$$\mathcal{C}_u(L)=\{(K_1,K_2):K_1,K_2,K_1+K_2\in\mathcal{C}(L)\},$$
where the letter $u$ in the subscript refers to union.
Finally we let $\mathcal{C}_0(L)$ denote the set of all strings appearing in all languages in $\mathcal{C}(L)$ or equivalently
$$\mathcal{C}_0(L)=\{s\in\Sigma^+:\{s\}\in\mathcal{C}(L)\}.$$

We refer to the following proposition a few times in this note. It is trivially followed from the definition.
\begin{proposition}\label{proposition:restriction}
Let $L$ be a language and let $K\in\mathcal{C}(L)$. The followings hold 
\begin{align*}
\mathcal{C}_0(K)&\subseteq \mathcal{C}_0(L),\\\mathcal{C}(K)&\subseteq \mathcal{C}(L),\\\mathcal{C}_c(K)&\subseteq \mathcal{C}_c(L),\\\mathcal{C}_u(K)&\subseteq \mathcal{C}_u(L).
\end{align*}
\end{proposition}

\begin{example}\label{ex:closure 2} It is evident that
	\begin{align*}
		\mathcal{C}_0(\Sigma^n) = \Sigma^{\leq n},
	\end{align*}
	whereas
	\begin{align*}
		\mathcal{C}_0(B(n,k)) = \{s\in\Sigma^{\leq n}: |s|_1 \leq k\}.
	\end{align*}
	We can also verify the followings
	\begin{align*}
	\mathcal{C}_c(\Sigma^n) &= \{(K_1,K_2): \emptyset \subset K_1\subseteq \Sigma^{n_1},\emptyset \subset K_2\subseteq \Sigma^{n_2}, n_1,n_2>0, n_1+n_2\leq n\}\\
	\mathcal{C}_{c}(B(n,k)) &= \{(K_1,K_2): \emptyset \subset K_1\subseteq B(n_1,k_1),\emptyset \subset K_2\subseteq B(n_2,k_2), n_1+n_2\leq n, k_1+k_2\leq k\}.
	\end{align*}
	We introduce the following shorthand notation which is helpful when working with binomial languages
	\begin{align*}
	\mathcal{C}_c(n,k) &= \{(n_1,k_1,n_2,k_2):0<n_1<n,0<n_2<n,0\leq k_1\leq \min(n_1,k),0\leq k_2\leq \min(n_2,k),\\ &\quad\quad n_1+n_2\leq n,k_1+k_2\leq k\},
	\end{align*}
	so we can for example write
	\begin{align*}
	\mathcal{C}_{c}(B(n,k)) &= \{(K_1,K_2): \emptyset \subset K_1\subseteq B(n_1,k_2),\emptyset \subset K_2\subseteq B(n_2,k_2), (n_1,k_1,n_2,k_2)\in \mathcal{C}_c(n,k)\}.
	\end{align*}
\end{example}


For a language $L\subseteq\{0,1\}^n$ a Boolean function is naturally associated. This is the Boolean function $f_L:\{0,1\}^n\rightarrow \{0,1\}$ defined such that $f_L^{-1}(1)=L$.  

\begin{example}[threshold languages]
    Let $\Sigma=\{0,1\}$ be a binary alphabet. Consider the omit language $$T(n,1) = \{s\in\Sigma^n:|s|_1\geq 1\}=\Sigma^n\setminus\{0^n\}.$$ For example we have $$T(3,1)=\{001,010,100,011,101,110,111\}.$$
     It is evident that $f_{T(n,1)}$ is the Boolean function $\OR_n$. More generally there are threshold languages $T(n,k)=\{s\in\Sigma^n:|s|_1\geq k\}$ defined for all $n > 0$ and $0\leq k\leq n$. Again when we write $T(n,k)$ it is implicitly assumed that $n,k$ satisfy the constraints $n > 0$ and $0\leq k\leq n$. 
\end{example}

Next we review the basics of linear programming. First let $\Sigma$ be an alphabet. The vector space $\real^\Sigma$ is naturally associated with this alphabet. We use $\real^+$ to denote the set of nonnegative real numbers. We let ${\real^{+}}^\Sigma$ denote the set of vectors in $\real^\Sigma$ with nonnegative entries. 

Let $\Sigma$ and $\Gamma$ be alphabets, let $A: \real^\Sigma\rightarrow\real^\Gamma$ be a linear transformation, and let $b\in \real^\Sigma$ and $c\in\real^\Gamma$ be vectors. A \emph{linear program} is a triple $(A,b,c)$ for which the following optimization problems are associated
\begin{equation*}
\openup\jot 
\begin{aligned}[t]
\mathcal{P}\quad\text{ maximize:}\quad &b^T x\\
        \text{subject to:}\quad & Ax\leq c,\\
        & x \in {\real^{+}}^\Sigma,
\end{aligned}
\qquad\qquad 
\begin{aligned}[t]
\mathcal{D}\quad\text{ minimize:}\quad &c^T y\\
        \text{subject to:}\quad & A^T y\geq  b,\\
        & y \in {\real^{+}}^\Gamma.
\end{aligned}
\end{equation*}
The problem $\mathcal{P}$ is called the \emph{primal problem} and the problem $\mathcal{D}$ is called the \emph{dual problem}. Let $\{a_i: i\in \Gamma\}\subset \real^\Sigma$ be the rows of $A$. An equivalent way of writing $\mathcal{P}$ and $\mathcal{D}$ is as follows
\begin{equation*}
\openup\jot 
\begin{aligned}[t]
\mathcal{P}\quad\text{ maximize:}\quad &b^T x\\
        \text{subject to:}\quad & a_i^Tx\leq c_i \quad\text{ for all } i \in \Gamma,\\
        & x \in {\real^{+}}^\Sigma,
\end{aligned}
\qquad\qquad 
\begin{aligned}[t]
\mathcal{D}\quad\text{ minimize:}\quad &c^T y\\
        \text{subject to:}\quad & \sum_{i\in\Gamma}a_i y_i\geq  b,\\
        & y \in {\real^{+}}^\Gamma.
\end{aligned}
\end{equation*}

\begin{example}
Let $\Sigma = \Gamma = \{1,2\}$. Fixing the natural ordering $1<2$ on this alphabet we can represent tuples $(x_1,x_2)$ in $\real^\Sigma$ by column vectors $\begin{bmatrix}x_1\\x_2\end{bmatrix}$. Let
$$
A=\begin{bmatrix}1 & -1\\1 & 2\end{bmatrix}, b = \begin{bmatrix}1 \\ 0\end{bmatrix}c=\begin{bmatrix}0\\3\end{bmatrix}.$$
The primal problem associated with the linear program $(A,b,c)$ is
\begin{align*}
\text{ maximize:}\quad &x_1\\
        \text{subject to:}\quad & x_1 - x_2 \leq 0,\\
        & x_1+2x_2 \leq 3,\\
        & x_1,x_2 \geq 0,
\end{align*}
and the dual problem is 
\begin{align*}
\text{ minimize:}\quad &3y_2\\
        \text{subject to:}\quad & y_1\begin{bmatrix}1\\-1\end{bmatrix} + y_2 \begin{bmatrix}1\\2\end{bmatrix}\geq \begin{bmatrix}1\\0\end{bmatrix},\\
        & y_1,y_2\geq 0.
\end{align*}

\end{example}

We say that $\mathcal{P}$ (resp. $\mathcal{D}$) is \emph{feasible} if there exists $\overline{x}$ (resp. $\overline{y}$) satisfying all the constraints in $\mathcal{P}$ (resp. $\mathcal{D}$) in which case we say that $\overline{x}$ is a primal feasible solution (resp. $\overline{y}$ is a dual feasible solution). If there exists no such solution we say that $\mathcal{P}$ (resp. $\mathcal{D}$) is \emph{infeasible}.  

We say that $\mathcal{P}$ is \emph{unbounded} if there exists a sequence of primal feasible solutions $x_1,x_2,\ldots$ such that $b^T x_n \to +\infty$ as $n\to +\infty$. Similarly we say that $\mathcal{D}$ is \emph{unbounded} if there exists a sequence of dual feasible solutions $y_1,y_2,\ldots$ such that $c^T y_n \to -\infty$ as $n\to +\infty$. 

We say that $\mathcal{P}$ is \emph{optimal} if there exists a feasible solution $\overline{x}$ such that for all feasible solutions $\hat{x}$ it holds that $b^T \overline{x} \geq b^T \hat{x}$. In such a case we say that $\overline{x}$ is an \emph{optimal solution} and the optimum value of $\mathcal{P}$, denoted by $\opt(\mathcal{P})$, is $b^T\overline{x}$. Similarly we can define optimal dual problems by replacing $\geq$ by $\leq$ in the definition of optimal primal problems.

\begin{theorem}[Fundamental Theorem of Linear Programming] Any optimization problem in the form of $\mathcal{P}$ or $\mathcal{D}$ is either infeasible, unbounded or optimal.
\end{theorem}

Suppose $\overline{x}$ and $\overline{y}$ are arbitrary primal and dual feasible solutions respectively. Then it holds that $b^T \overline{x} \leq c^T \overline{y}$. This is known as \emph{weak duality}. If in addition it holds that $b^T \overline{x} = c^T \overline{y}$, then $\overline{x}$ and $\overline{y}$ are both optimal in their respective problems. In our example $\overline{x} = \begin{bmatrix}1\\1\end{bmatrix}$ and $\overline{y} = \begin{bmatrix}\frac{2}{3}\\\frac{1}{3}\end{bmatrix}$ are primal and dual feasible solutions respectively and their objective values are the same. Therefore from weak duality both $\overline{x}$ and $\overline{y}$ are optimal solutions. 

\begin{theorem}[Duality Theorem] If both $\mathcal{P}$ and $\mathcal{D}$ have feasible solutions, then they both have optimal solutions and their optimum values are the same. 
\end{theorem}

In this note we associate linear programs to languages. The alphabets $\Sigma$ and $\Gamma$ in the definition of these linear programs are often the sets $\mathcal{C}_0(L)$ or $\mathcal{C}(L)$ for some language $L$.

\section{Notations}\label{section:notations}
The following notations are fixed in the rest of this note. 
We reserve the letter $R$ for regular expressions, the letter $a$ for alphabet symbols, and the letters $r,s,t$ for strings. Letters $B$ and $T$ denote binomial and threshold languages respectively. Letters $K,L$ are also reserved for languages whereas letters $E,I,J$ denote special vectors we define later. 

Letters $i,j,k,l,m,n,p,q$ always denote nonnegative (and most often positive) integers. Letters $\alpha$ and $\beta$ are reserved for real numbers. The set of integers is denoted by $\mathbb{N}$. Letter $g$ always denotes a map $\Sigma^+\rightarrow \real^{+}$ whereas $f$ denotes Boolean functions $\Sigma^n\rightarrow \{0,1\}$. The natural logarithm is denoted by $\ln$, whereas the binary logarithm is denoted by $\log$. The base of the natural logarithm is denoted by $e=2.7182\cdots$.

Letters $w,x,y,z,W,X,Y,Z$ (usually with a string or a language in the subscript) are reserved for variables appearing in optimization problems. The two scripted letters $\mathcal{P}$ and $\mathcal{D}$ refer to primal and dual problems. We associate a linear program to a language, for example we often write $\mathcal{P}(L)$. 

\section{Linear programming formulations}\label{section:formulations}
    We present two linear programming formulations associated with languages. We start with the \emph{weak formulation} and later in the section we define the \emph{strong formulation}. The use of the words weak and strong becomes clear later in the section.
    
    Let $\Sigma$ be an alphabet and $L$ be a language over $\Sigma$. We associate an optimization problem $\mathcal{P}(L)$ to $L$ as follows
    \begin{align*}
        \mathcal{P}(L)\quad\text{ maximize:}\quad &\sum_{s\in L}x_s\\
        \text{subject to:}\quad & \sum_{s\in K_1K_2}x_s \leq \sum_{s\in K_1}x_s + \sum_{s\in K_2}x_s  \quad\text{ for all } (K_1,K_2)\in \mathcal{C}_c(L),\\
        &0 \leq x_s \leq |s|\quad\text{ for all } s \in \mathcal{C}_0(L).
    \end{align*}
    This is evidently a linear programming problem in the primal form. The dual of $\mathcal{P}(L)$ is calculated in the next section. The problem is feasible; setting $x_s=0$ for all $s\in\mathcal{C}_0(L)$ is a trivially feasible solution. It is also bounded, as seen by the constraints $x_s\leq |s|$. So by the fundamental theorem of linear programming there exists an optimal solution for $\mathcal{P}(L)$. 
    \begin{example}\label{ex:simple primal}
    Let $\Sigma=\{0\}$ be a unary alphabet, and consider the language $L=\{00,000\}\subset\Sigma$. The closure is evidently 
    \begin{align*}
    \mathcal{C}(L)=\{0,00,000,0+00,00+000\},
    \end{align*} and 
    \begin{align*}
    \mathcal{C}_0(L)&=\{0,00,000\},\\
    \mathcal{C}_c(L) &= \{(0,0),(0,00),(00,0),(0+00,0),(0,0+00)\}.
    \end{align*} Consequently the problem $\mathcal{P}(L)$, after removing the redundant constraints, is as follows
    \begin{align*}
        \quad\text{ minimize:}\quad & x_{00}+x_{000}\\
        \text{subject to:}\quad &x_{00}\leq  2x_0,\\
        &x_{000}\leq x_0+x_{00},\\
        &x_{00}+x_{000}\leq 2x_{0}+x_{00},\\
        &x_0\leq 1,\\
        &x_{00}\leq 2,\\
        &x_{000}\leq 3,\\
        &x_0,x_{00},x_{000}\geq 0.
    \end{align*}
    We can verify that $\overline{x}_0=1,\overline{x}_{00}=\overline{x}_{000}=2$ is a feasible solution, and its objective value is $4$.
    \end{example}
    
    One caveat of this formulation is its size. Luckily we are not looking for an efficient algorithm for solving $\mathcal{P}(L)$.\footnote{In fact an efficient algorithm should not exist if we accept basic complexity-theoretic assumptions as explained in Section \ref{section:complexity}.} Instead we use this formulation to prove lower bounds on the length of optimal regular expressions of $L$. 
    We see another caveat of this formulation related to proving lower bounds in Example \ref{ex:caveat}. For now we prove that the objective value of any feasible solution of $\mathcal{P}(L)$ is a lower bound on the length of any regular expression of $L$.
    \begin{theorem}\label{thm:lowerbound}
        Let $\overline{x}$ be a feasible solution of $\mathcal{P}(L)$  and let $R$ be a regular expression of $L$. It holds that
        \begin{align*}
            \sum_{s\in L}\overline{x}_s \leq |R|.
        \end{align*}
        In particular it holds that $\opt(\mathcal{P}(L))\leq |R|$.
    \end{theorem}
    \begin{proof}
        Let $R$ be a regular expression of $L$. The proof is by induction on $|R|$. When $R$ is an alphabet letter the proof is evident from the observation that $\overline{x}_a\leq 1$ for all $a\in\Sigma$ as $\overline{x}$ is a feasible solution. Now suppose $R$ is not a single term; so it is either $R=(R_1+R_2)$ or $R=R_1R_2$. In both of these cases, the restriction of $\overline{x}$ to $\mathcal{C}_0(L(R_1))$ is a feasible solution of $\mathcal{P}(L(R_1))$. This follows from Proposition \ref{proposition:restriction} and the observation that $L(R_1)\in \mathcal{C}(L)$. So by induction hypothesis it holds that $\sum_{s\in L(R_1)} \overline{x}_s\leq |R_1|$. Similarly it holds that $\sum_{s\in L(R_2)} \overline{x}_s\leq |R_2|$. When $R=(R_1+R_2)$, by feasibility of $\overline{x}$ we can write $$\sum_{s\in L}\overline{x}_s=\sum_{s\in L(R_1)\cup L(R_2)}\overline{x}_s \leq \sum_{s\in L(R_1)}\overline{x}_s + \sum_{s\in L(R_2)}\overline{x}_s \leq |R_1|+|R_2|=|(R_1+R_2)|=|R|$$ which completes the proof for this case. Now suppose $R=R_1R_2$ and observe that $(L(R_1),L(R_2))\in \mathcal{C}_c(L)$. So from feasibility of $\overline{x}$ we can write
        \begin{align*}
			\sum_{s\in L} \overline{x}_s = \sum_{s\in L(R_1)L(R_2)} \overline{x}_s \leq \sum_{s\in L(R_1)} \overline{x}_s + \sum_{s\in L(R_2)} \overline{x}_s\leq |R_1|+|R_2|=|R_1R_2|=|R|.
        \end{align*}
    \end{proof}
    Recall the language $L$ in Example \ref{ex:simple primal}. In the light of Theorem \ref{thm:lowerbound} and the feasible solution for $\mathcal{P}(L)$ given in that example, the regular expression $(0+00)0$ is optimal for $L$.
    
    \begin{example}\label{ex:all strings primal} Let $\Sigma=\{0,1\}$ be the binary alphabet and consider the language $L=\Sigma^n$ for some $n > 0$. A regular expression for this language is $R=(0+1)^n$ which is of length $2n$. We give a feasible solution $\overline{x}$ of $\mathcal{P}(\Sigma^n)$ with an objective value of $2n$, and by doing so we prove that $R$ is optimal. First recall Examples \ref{ex:closure 2} where we calculated $\mathcal{C}_0(\Sigma^n)$ and $\mathcal{C}_c(\Sigma^n)$. Let $\overline{x}_s = \frac{|s|}{|\Sigma|^{|s|-1}}$ for all $s\in\Sigma^{\leq n}=\mathcal{C}_0(\Sigma^n)$. The objective value of $\overline{x}$ is easily seen to be $|\Sigma^n|\frac{n}{|\Sigma|^{n-1}} = |\Sigma|n=2n$. Next we verify that the constraints in $\mathcal{P}(\Sigma^n)$ are satisfied. It holds trivially that $0\leq\overline{x}_s\leq |s|$ for all $s\in\Sigma^{\leq n}$. Now suppose $(K_1,K_2)\in \mathcal{C}_c(\Sigma^n)$. From Example \ref{ex:closure 2}, there exist $n_1$ and $n_2$ such that $0 < n_1,n_2< n, n_1+n_2\leq n$ and $K_1\subseteq \Sigma^{n_1}$ and $K_2\subseteq \Sigma^{n_2}$. On one hand we have $$\sum_{s\in K_1K_2} \overline{x}_s = |K_1||K_2|\frac{n_1+n_2}{|\Sigma|^{n_1+n_2-1}},$$
    and on the other hand we have $$\sum_{s\in K_1} \overline{x}_s = |K_1|\frac{n_1}{|\Sigma|^{n_1-1}}$$ and $$\sum_{s\in K_2} \overline{x}_s = |K_2|\frac{n_2}{|\Sigma|^{n_2-1}}.$$ So we must show that $$|K_1||K_2|\frac{n_1+n_2}{|\Sigma|^{n_1+n_2-1}}\leq |K_1|\frac{n_1}{|\Sigma|^{n_1-1}} + |K_2|\frac{n_2}{|\Sigma|^{n_2-1}}$$ or equivalently
    $$\frac{n_1+n_2}{|\Sigma|^{n_1+n_2-1}}\leq \frac{1}{|K_2|}\frac{n_1}{|\Sigma|^{n_1-1}} + \frac{1}{|K_1|}\frac{n_2}{|\Sigma|^{n_2-1}}.$$ Since the right hand size is minimized when $K_1=\Sigma^{n_1}$ and $K_2=\Sigma^{n_2}$, it suffices to prove the inequality only for this case. In this case the inequality turns into an equality. This finishes the proof of the claim that $\overline{x}$ is a feasible solution.
    \end{example}
    
    \begin{example}\label{ex:binomial} 
    	Let $\Sigma=\{0,1\}$ be the binary alphabet and consider the binomial language $B(n,1)$ for some $n > 0$. A regular expression $R_n$ for this language is given in \cite{ellul-regular-expressions} recursively by letting $R_1 = 1$ and $$R_n = \Bigl(0^{\lfloor\frac{n}{2}\rfloor}R_{\lceil\frac{n}{2}\rceil}+R_{\lfloor\frac{n}{2}\rfloor}0^{\lceil\frac{n}{2}\rceil}\Bigr)$$ for all $n > 1$. They also showed that $|R_n|=\lceil n\log(2n)\rceil$. Using Theorem \ref{thm:lowerbound} we show that $R_n$ is asymptotically optimal. We do so by arguing that $\overline{x}$ defined as follows is a feasible solution of  $\mathcal{P}({B(n,1)})$. First recall Example \ref{ex:closure 2} where we calculated $\mathcal{C}_0(B(n,1))$ and $\mathcal{C}_c(B(n,1))$. For all $s\in\Sigma^{\leq n}$ for which $|s|_1\leq 1$ let
    	\begin{align*}
    	\overline{x}_s=\begin{cases} 
    	|s| & |s|_1=0, \\
    	\ln(e|s|) & |s|_1=1.
    	\end{cases}
    	\end{align*}
    	The presence of the constant $e$ here is not conveying any significant truth (except that it looks nicer to have $\overline{x}_0=\overline{x}_1 =1$).
    	Let $(K_1, K_2)\in\mathcal{C}_c(B(n,1))$. We must prove the inequality $$\sum_{s\in K_1K_2}\overline{x}_s \leq \sum_{s\in K_1}\overline{x}_s + \sum_{s\in K_2}\overline{x}_s.$$ First note that from Example \ref{ex:closure 2} there exist $(n_1,k_1,n_2,k_2)\in \mathcal{C}_c(n,k)$ such that $K_1\subseteq B(n_1,k_1)$  and $K_2\subseteq B(n_2,k_2)$. There are three cases to consider; we have that $(k_1,k_2)$ is either $(0,0), (1,0)$ or $(0,1)$. The simplest case is when $k_1=0$ and $k_2=0$. For this case the inequality holds trivially. Now suppose $k_1=1$ and $k_2 = 0$. For this case the inequality is equivalent to 
    	\begin{align}\label{eq:example:binomial:inequality}
    	|K_1|\ln(e(n_1+n_2))\leq |K_1|\ln(en_1) + n_2,
    	\end{align}
    	which can also be written as
    	$$|K_1|\ln(1+\frac{n_2}{n_1})\leq n_2.$$
    	The inequality follows with the observations that $\ln(1+x)\leq x$ for all $x\geq 0$ and also that $|K_1|\leq n_1$. The remaining case of $k_1=0$ and $k_2=1$ is symmetric to this one. Finally note that the objective value of $\overline{x}$ is $\sum_{s\in B(n,1)}\overline{x}_s = n\ln(en)$. 
    \end{example}
    \begin{problem} What is the quantity $\opt(P(B(n,1)))$? We conjecture that it equals $\lceil n\log(2n)\rceil$. 
    	
        We verified this conjecture for $n\leq 16$ using the CVXPY \cite{cvxpy} library. The optimal solution obtained for $n=8$ is as follows (showing only 2 decimal places):
    	\begin{center}
    		\begin{tabular}{ c c c c c c c c c c }
    			0   &1.00   &0000  &4.00    &01000  & 3.14 &0000000 & 7.00    & 00000000 & 8.00\\ 
    			1   &1.00   &0001  &3.11 &10000  & 3.73 &0000001 & 4.54 & 00000001 & 4.93\\
    			00  &2.00   &0010  &2.88 &000000 & 6.00    &0000010 & 3.83 & 00000010 & 4.16\\
    			01  &2.00   &0100  &2.88 &000001 & 4.15 &0000100 & 3.39 & 00000100 & 3.59\\
    			10  &2.00   &1000  &3.11 &000010 & 3.51 &0001000 & 3.22 & 00001000 & 3.30\\
    			000 &3.00   &00000 &5.00    &000100 & 3.16 &0010000 & 3.39 & 00010000 & 3.30\\
    			001 &2.92&00001 &3.73 &001000 & 3.16 &0100000 & 3.83 & 00100000 & 3.59\\
    			010 &2.07&00010 &3.14 &010000 & 3.51 &1000000 & 4.54 & 01000000 & 4.16\\
    			100 &2.92&00100 &2.96 &100000 & 4.15 &        &      & 10000000 & 4.93\\
    		\end{tabular}
    	\end{center}
    \end{problem}
    In fact we conjecture that there is no gap between $\opt(\mathcal{P}(B(n,k)))$ and the length of optimal regular expressions of $B(n,k)$ for all values of $k$. See Open Problem \ref{problem:collapse}. This is not true in general as shown next.
    
    Our next example demonstrates that for some languages the weak formulation fails to prove a good lower bound. In the example we present a family of languages $\{L_n: n > 0\}$ for which the gap between $\opt(\mathcal{P}(L_n))$ and the length of optimal regular expressions of $L_n$ is unbounded as $n$ grows. 
    \begin{example}\label{ex:caveat}
    	Recall the threshold languages $T(n,1)$. Again from \cite{ellul-regular-expressions} we have that the relation $$T(n,1)=\Sigma^{\lfloor\frac{n}{2}\rfloor}T(\lceil\frac{n}{2}\rceil,1)+T(\lfloor\frac{n}{2}\rfloor,1)\Sigma^{\lceil\frac{n}{2}\rceil},$$ provides us with a regular expression of $T(n,1)$ of length $\lceil2n\log(2n)\rceil$. We see in Proposition \ref{prop:threshold} that this is asymptotically optimal.
    	On the other hand, we show that $$\opt(\mathcal{P}(T(n,1))) = O(n).$$
    	Suppose $\overline{x}$ is a feasible solution of $\mathcal{P}(T(n,1))$. It holds that
    	\begin{align*}
    	\sum_{s\in T(n,1)} \overline{x}_s &= \sum_{s\in \Sigma^{\lfloor\frac{n}{2}\rfloor}T(\lceil\frac{n}{2}\rceil,1)+T(\lfloor\frac{n}{2}\rfloor,1)\Sigma^{\lceil\frac{n}{2}\rceil}}\overline{x}_s\\&\leq \sum_{s\in \Sigma^{\lfloor\frac{n}{2}\rfloor}T(\lceil\frac{n}{2}\rceil,1)}\overline{x}_s +  \sum_{s\in T(\lfloor\frac{n}{2}\rfloor,1)\Sigma^{\lceil\frac{n}{2}\rceil}}\overline{x}_s\\
    	&\leq \sum_{s\in \Sigma^{\lfloor\frac{n}{2}\rfloor}}\overline{x}_s + 
    	\sum_{s\in T(\lceil\frac{n}{2}\rceil,1)}\overline{x}_s
    	+\sum_{s\in T(\lfloor\frac{n}{2}\rfloor,1)}\overline{x}_s+ \sum_{s\in \Sigma^{\lceil\frac{n}{2}\rceil}}\overline{x}_s\\
    	&\leq 2\sum_{s\in\Sigma^{ \lfloor \frac{n}{2}\rfloor}}\overline{x}_s + 2\sum_{s\in\Sigma^{\lceil\frac{n}{2}\rceil}}\overline{x}_s\\
    	&\leq 2\opt(\mathcal{P}(\Sigma^{ \lfloor \frac{n}{2}\rfloor})) + 2\opt(\mathcal{P}(\Sigma^{ \lceil \frac{n}{2}\rceil}))\\
    	&= 4\lfloor\frac{n}{2}\rfloor + 4\lceil\frac{n}{2}\rceil = 4n.
    	\end{align*}
    	To obtain the second to last inequality we used the fact that $$\Sigma^{ \lfloor \frac{n}{2}\rfloor}=T(\lfloor\frac{n}{2}\rfloor,1) + 0^{ \lfloor \frac{n}{2}\rfloor}$$ thus
    	$$\sum_{s\in T(\lfloor\frac{n}{2}\rfloor,1)}\overline{x}_s \leq \sum_{s\in \Sigma^{\lfloor\frac{n}{2}\rfloor}}\overline{x}_s.$$
    	Similarly it holds that
    	$$\sum_{s\in T(\lceil\frac{n}{2}\rceil,1)}\overline{x}_s \leq \sum_{s\in \Sigma^{\lceil\frac{n}{2}\rceil}}\overline{x}_s.$$
    	
    	 For the last inequality we used the fact that $\overline{x}$ when restricted to $\Sigma^{\leq \lfloor\frac{n}{2}\rfloor}$ is a feasible solution of $\mathcal{P}(\Sigma^{\lfloor\frac{n}{2}\rfloor})$. This follows from Proposition \ref{proposition:restriction} and the observation that $\Sigma^{\lfloor\frac{n}{2}\rfloor} \in \mathcal{C}(T(n,1))$. Therefore it holds that $$\sum_{s\in\Sigma^{ \lfloor \frac{n}{2}\rfloor}}\overline{x}_s \leq \opt(\mathcal{P}(\Sigma^{ \lfloor \frac{n}{2}\rfloor})),$$ and similarly $$\sum_{s\in\Sigma^{ \lceil \frac{n}{2}\rceil}}\overline{x}_s \leq \opt(\mathcal{P}(\Sigma^{ \lceil \frac{n}{2}\rceil})).$$ We showed in Example \ref{ex:all strings primal} that $\opt(\mathcal{P}(\Sigma^{\lfloor\frac{n}{2}\rfloor})) = 2\lfloor\frac{n}{2}\rfloor$ and $\opt(\mathcal{P}(\Sigma^{\lceil\frac{n}{2}\rceil})) = 2\lceil\frac{n}{2}\rceil$. 
    \end{example}
    
    We now give a stronger linear programming formulation for which the optimal value always coincides with the length of optimal regular expressions. The primal problem of the \emph{strong linear programming formulation} associated with a language $L$ is denoted by $\mathcal{P}_S(L)$ and is defined as follows:
    \begin{align*}
    \mathcal{P}_S(L)\quad\text{ maximize:}\quad &X_L\\
    \text{subject to:}\quad & X_{K_1K_2}\leq X_{K_1}+X_{K_2}  \quad\quad\text{ for all } (K_1,K_2)\in \mathcal{C}_c(L)\\
    & X_{K_1+K_2}\leq X_{K_1}+X_{K_2}  \quad\text{ for all } (K_1,K_2)\in \mathcal{C}_u(L),\\
    &X_s\leq |s|\quad\text{ for all }s\in\mathcal{C}_0(L),\\
    &X_K\geq 0\text{ for all } K\in\mathcal{C}(L).
    \end{align*}
    The subscripts of $X$ are languages in $\mathcal{C}(L)$, i.e., $X\in{\real^+}^{\mathcal{C}(L)}$. It should however be noted that for $s\in \mathcal{C}_0(L)$ we use $X_s$ as a shorthand for $X_{\{s\}}$.
    
    The problem $\mathcal{P}_S(L)$ is feasible; setting $X_K=0$ for all $K\in\mathcal{C}(L)$ is a trivially feasible solution. It is also bounded as seen by the constraints $X_s\leq |s|$ which together with the constraints $X_{K_1+K_2}\leq X_{K_1}+X_{K_2}$ imply that $X_L\leq \sum_{s\in L} |s|$ for every feasible solution $X$. So there must exist an optimal solution for this problem.
    Also note that every feasible solution $\overline{x}$ of $\mathcal{P}(L)$ can be turned into a feasible solution $\overline{X}$ of $\mathcal{P}_S(L)$ by letting $\overline{X}_K = \sum_{s\in K} \overline{x}_s$. Moreover the objective values of $\overline{X}$ and $\overline{x}$ are the same. Therefore it holds that $$\opt(\mathcal{P}(L))\leq \opt(\mathcal{P}_S(L)).$$

    \begin{theorem}\label{thm:strong} It holds that $\opt(\mathcal{P}_S(L))$ equals the length of optimal regular expressions of $L$. 
    \end{theorem}
    \begin{proof}
    	Let $\Sigma$ be an alphabet, let $R$ be a regular expression over $\Sigma$, let $L=L(R)$, and let $\overline{X}$ be any feasible solution of $\mathcal{P}_S(L)$. First we prove that $\overline{X}_L \leq |R|$. The proof is by induction on $|R|$. The case of $R=a\in\Sigma$ follows from the observation that $\overline{X}$ must satisfy $\overline{X}_{a}\leq 1$. So assume that $|R| > 1$. So there must exist regular expressions $R_1$ and $R_2$ for which either $R=R_1R_2$ or $R=(R_1+R_2)$. First suppose that $R=R_1R_2$. Let $L_1=L(R_1)$ and $L_2=L(R_2)$. Since $(L_1,L_2)\in \mathcal{C}_c(L)$ and $\overline{X}$ is feasible, it must hold that $\overline{X}_{L}\leq \overline{X}_{L_1} + \overline{X}_{L_2}$. Now $\overline{X}$ restricted to $\mathcal{C}(L_1)$ is a feasible solution for $\mathcal{P}_S(L_1)$, so from the induction hypothesis it must hold that $\overline{X}_{L_1}\leq |R_1|$. Similarly it must hold that $\overline{X}_{L_2}\leq |R_2|$. So we can write $\overline{X}_{L}\leq \overline{X}_{L_1}+\overline{X}_{L_2}\leq |R_1|+|R_2|=|R_1R_2|=|R|$. The case of $R=(R_1+R_2)$ can be handled similarly. Since this is true for all feasible solutions we proved that $\opt(\mathcal{P}_S(L))\leq |R|$.
    	
    	Now suppose $L$ is a language and $R$ is an optimal regular expression of $L$.  We show that $\opt(\mathcal{P}_S(L)) \geq |R|$.
    	For any $K\in \mathcal{C}(L)$ let $R_K$ denote an optimal regular expression of $K$. Define $\overline{X}$  by setting $\overline{X}_{K} = |R_K|$ for all $K\in\mathcal{C}(L)$. We show that $\overline{X}$ is a feasible solution of $\mathcal{P}_S(L)$. Suppose that $(K_1, K_2)\in\mathcal{C}_c(L)$. Since $R_{K_1K_2}$ is an optimal regular expression of $K_1K_2$, and $R_{K_1}R_{K_2}$ is a regular expression of $K_1K_2$, it must hold that $|R_{K_1K_2}|\leq |R_{K_1}R_{K_2}| = |R_{K_1}|+|R_{K_2}|$. Thus we have $\overline{X}_{K_1K_2}\leq \overline{X}_{K_1}+\overline{X}_{K_2}$. We can similarly show that $\overline{X}_{K_1+K_2}\leq \overline{X}_{K_1} + \overline{X}_{K_2}$ for all $(K_1,K_2)\in\mathcal{C}_u(L)$. This completes the proof that $\overline{X}$ is feasible. Now notice that the objective value of $\overline{X}$ is $\overline{X}_L=|R_L|=|R|$. Therefore it must hold that $\opt(\mathcal{P}_S(L)) \geq \overline{X}_L = |R|$.
    \end{proof}
    
    Example \ref{ex:caveat} shows that $\mathcal{P}(T(n,1))$ fails to prove any nontrivial lower bounds on $T(n,1)$. The next proposition proves a matching lower bound for $T(n,1)$ using $\mathcal{P}_S(T(n,1))$.
    We saw earlier that the Boolean function $f_{T(n,1)}$ is $\OR_n$, i.e., the disjunction of $n$ bits. The function value of $\OR_n$ is sensitive to the bits of the all zero input in the sense that flipping any one bit in $0^n$ flips the function value. Similarly the function value is sensitive on all inputs $s$ for which $|s|_1=1$ in that it is possible to change one bit and change the function value. The function value is more robust on all the other inputs (one must change at least two bits in order to flip the function value). In an intuitive level one can blame the complexity of a Boolean function on inputs for which the function value is most sensitive on the input bits. This intuitive idea is used to prove quantum query lower bounds for the Boolean function $f_{T(n,k)}$. We use this idea to prove a matching lower bound for the length of regular expressions of $T(n,1)$ in the next proposition.
    
    \begin{proposition}\label{prop:threshold} It holds that $$\opt(\mathcal{P}_S(T(n,1))) \geq n\ln(en).$$ Therefore the regular expression of $T(n,1)$ given in Example \ref{ex:caveat} is optimal up to a constant multiplicative factor.
    \end{proposition}
    \begin{proof}
    	We construct a feasible solution $\overline{X}$ for $\mathcal{P}_S(T(n,1))$ using the feasible solution $\overline{x}$ for $\mathcal{P}(B(n,1))$ given in Example \ref{ex:binomial}. Let $K\in\mathcal{C}(T(n,1))$ be a language. All strings in $K$ must have the same length. This follows from Proposition \ref{proposition:restriction} and the fact that for every language in $\mathcal{C}(\Sigma^n)$ all strings are of the same length. Suppose that the length of strings in $K$ is $n'$. If $0^{n'}\in K$ then set $\overline{X}_K = n'$. Otherwise set $$\overline{X}_K = \sum_{s\in B(n',1)\cap K}\overline{x}_s$$ where $\overline{x}$ is defined in Example \ref{ex:binomial}. In other words, in the latter case $\overline{X}_K = m\ln(en')$ where $m$ is the number of strings in $K$ that contain exactly one $1$. The objective value of $\overline{X}$ is thus $\overline{X}_{T(n,1)}=n\ln(en)$.
    	
    	Next we show that $\overline{X}$ is feasible. The positivity constraints are trivially satisfied, so are the boundedness constraint $\overline{X}_s \leq |s|$. 
    	
    	Let $(K_1,K_2)\in\mathcal{C}_c(T(n,1))$. We show that $\overline{X}_{K_1K_2}\leq \overline{X}_{K_1}+\overline{X}_{K_2}$. If both $K_1$ and $K_2$ contain strings in $0^+$ then so is $K_1K_2$ in which case the left hand side of the inequality is zero (and the inequality holds trivially). If neither $K_1$ nor $K_2$ contain strings in $0^+$, then $K_1K_2$ contains no string with exactly one $1$. Thus $X_{K_1K_2}=0$ and the inequality holds trivially again. Now suppose only one of $K_1$ or $K_2$ contains a string in $0^+$. Either way the inequality becomes identical to the inequality (\ref{eq:example:binomial:inequality}) that we proved in Example \ref{ex:binomial}.
    	
    	Now let $(K_1,K_2)\in\mathcal{C}_u(T(n,1))$. We show that $\overline{X}_{K_1+K_2}\leq \overline{X}_{K_1}+\overline{X}_{K_2}$. If any of $K_1$ or $K_2$ contains a string in $0^+$, then $\overline{X}_{K_1+K_2}=0$ and the inequality holds trivially. So suppose neither $K_1$ nor $K_2$ contain a string in $0^+$. Let $m_1,m_2$ be the number of strings in $K_1$ and $K_2$ respectively with exactly one $1$. The inequality then follows from the observation that the number of strings in $K_1+K_2$ with exactly one $1$ is at most $m_1+m_2$. 
    \end{proof}
    
    The strategy in this proof can be extended to obtain feasible solutions of $\mathcal{P}_S(T(n,k))$ from feasible solutions of $\mathcal{P}(B(n,k))$ with the same objective value for all $k$. Therefore lower bounds on $B(n,k)$ obtained from feasible solutions of $\mathcal{P}(B(n,k))$ are also lower bounds on $T(n,k)$.
    
    We end this section with a digression by noting a simple property that holds for the strong formulation but not necessarily for the weak formulation. Let $L\subseteq \Sigma^n$ be a language and define
    \begin{align*}
    		\overline{\mathcal{P}}_S(L)\quad\text{ maximize:}\quad &X_L\\
    		\text{subject to:}\quad & X_{K_1K_2}\leq X_{K_1}+X_{K_2}  \quad\quad\text{ for all } (K_1,K_2)\in \mathcal{C}_c(\Sigma^n)\\
    		& X_{K_1+K_2}\leq X_{K_1}+X_{K_2}  \quad\text{ for all } (K_1,K_2)\in \mathcal{C}_u(\Sigma^n),\\
    		&X_s\leq |s|\quad\text{ for all }s\in\mathcal{C}_0(\Sigma^n),\\
    		&X_K\geq 0\text{ for all } K\in\mathcal{C}(\Sigma^n).
    \end{align*}
    We do not give a proof here but it holds that $\opt(\overline{\mathcal{P}}_S(L))=\opt(\mathcal{P}_S(L))$. This is an interesting observation because only the objective function of $\overline{\mathcal{P}}_S(L)$ is dependent on the language $L$. To put differently the set of feasible solutions of $\overline{\mathcal{P}}_S(L_1)$ and $\overline{\mathcal{P}}_S(L_2)$ are the same for every pair of languages $L_1,L_2\subseteq \Sigma^n$. 
    
    We note that the analogue of this does not hold for the weak formulation.  More formally, for a language $L\subseteq \Sigma^n$ define $\overline{\mathcal{P}}(L)$ as follows
    \begin{align*}
    \overline{\mathcal{P}}(L)\quad\text{ maximize:}\quad &\sum_{s\in L}x_s\\
    \text{subject to:}\quad & \sum_{s\in K_1K_2}x_s \leq \sum_{s\in K_1}x_s + \sum_{s\in K_2}x_s  \quad\text{ for all } (K_1,K_2)\in \mathcal{C}_c(\Sigma^n),\\
    &0 \leq x_s \leq |s|\quad\text{ for all } s \in \mathcal{C}_0(\Sigma^n).
    \end{align*}
    From Example \ref{ex:all strings primal} we know that $ \opt(\overline{\mathcal{P}}(\Sigma^n))=\opt(\mathcal{P}(\Sigma^n))=n|\Sigma|$. This implies that for every language $L\subseteq \Sigma^n$ it must be that $\opt(\overline{\mathcal{P}}(L))\leq n|\Sigma|$. On the other hand we showed in Example \ref{ex:binomial} that $\opt(\mathcal{P}(B(n,1)))\geq n\ln(2n)$.
    
\section{Dual problems}\label{section:duals}
    In this section we compute the duals of $\mathcal{P}(L)$ and $\mathcal{P}_S(L)$ denoted by $\mathcal{D}(L)$ and $\mathcal{D}_S(L)$ respectively. We show that the dual problems have nice connections to regular expressions of $L$; every regular expression of $L$ can be turned into a feasible solution of $\mathcal{D}(L)$ (and also $\mathcal{D}_S(L)$). Moreover the objective value of the dual feasible solution obtained in this way coincides with the length of the regular expression.
    
    To compute the dual we first transform $\mathcal{P}(L)$ into the formal definition in the introduction.
    Let $L$ be a language. For a language $L$ its indicator vector $I_L\in \mathbb{R}^{\mathcal{C}_0(L)}$ is defined to be
    \begin{align*}
    I_L(s) = \begin{cases}1 &\quad\text{ for all }s\in L,\\
                          0 &\quad\text{ otherwise.}\end{cases}
    \end{align*}
   Denote $I_{\{s\}}$ simply by $I_s$ for all $s\in \mathcal{C}_0(L)$. Also define a vector $J_L = \sum_{s\in\mathcal{C}_0(L)} |s|I_s$.
    The primal problem $\mathcal{P}(L)$ can be rewritten as follows
    \begin{align*}
        \mathcal{P}(L)\quad\text{ maximize:}\quad &I_L^Tx\\
        \text{subject to:}\quad &(I_{K_1K_2} - I_{K_1} - I_{K_2})^Tx \leq 0  \quad\text{ for all } (K_1,K_2)\in \mathcal{C}_c(L),\\
        & I_s^T x \leq |s| \quad\text{ for all } s\in \mathcal{C}_0(L),\\
        &x \in {\real^+}^{\mathcal{C}_0(L)}.
    \end{align*}
    Now the dual is easily seen to be
    \begin{align*}
        \mathcal{D}(L)\quad\text{ minimize:}\quad &J_L^Tw\\
        \text{subject to:}\quad &\sum_{(K_1,K_2)\in\mathcal{C}_c(L)}y_{K_1,K_2}(I_{K_1K_2} - I_{K_1} - I_{K_2}) + w \geq I_L,\\
        &y_{K_1,K_2} \geq 0 \quad\text{ for all }(K_1,K_2)\in\mathcal{C}_c(L),\\
        &w \in {\real^+}^{\mathcal{C}_0(L)}.
    \end{align*}
    
    The dual problem is always feasible. One way to see this is to use the duality theorem. We can also see this directly by setting $w_s = 1$ for all $s\in L$ and setting all the other variables to zero. The objective value of this trivial feasible solution is $\sum_{s\in L}|s|$. Evidently this is the length of a trivial regular expression for $L$ in which every string in $L$ is repeated exactly once as a term. Just like the primal, the dual problem is bounded too. This follows from the observation that the objective value of any feasible solution is nonnegative.  We use pairs of vectors $(w,y)\in {\real^+}^{\mathcal{C}_0(L)} \times {\real^+}^{\mathcal{C}_c(L)}$ to refer to dual feasible solutions.
    
    \begin{example}\label{ex:simple dual}
    Let $L=\{00,000\}$ be the language we studied in Example \ref{ex:simple primal}. The dual problem $\mathcal{D}(L)$ for this language is obtained easily if we recall that
    \begin{align*}
    \mathcal{C}_0(L)&=\{0,00,000\},\\
    \mathcal{C}_c(L) &= \{(0,0),(0,00),(00,0),(0+00,0),(0,0+00)\}.
    \end{align*}
    The dual is therefore as follows
    \begin{align*}
    \quad\text{ minimize:}\quad & w_0 + 2w_{00}+3w_{000}\\
        \text{subject to:}\quad 
        &y_{0,0}(I_{00}-I_0-I_0)+(y_{0,00}+y_{00,0})(I_{000}-I_0-I_{00})\\&\quad+(y_{0,0+00}+y_{0+00,0})(I_{00+000}-I_{0}-I_{0+00}) + w \geq I_{00+000}\\
        &y_{0,0},y_{0,00},y_{00,0},y_{0,0+00},y_{0+00,0} \geq 0\\
        &w_0,w_{00},w_{000} \geq 0.
    \end{align*}
    Noting $I_{00+000}=I_{00}+I_{000}$ and $I_{0+00}=I_{0}+I_{00}$ we can simplify to obtain the following
    \begin{align*}
    \quad\text{ minimize:}\quad & w_0 + 2w_{00}+3w_{000}\\
        \text{subject to:}\quad 
        &(-2y_{0,0}-y_{0,00}-y_{00,0}-2y_{0,0+00}-2y_{0+00,0}+w_0)I_0\\
         &\quad+(y_{0,0}-y_{0,00}-y_{00,0}+w_{00}-1)I_{00}\\
         &\quad+(y_{0,00}+y_{00,0}+y_{0,0+00}+y_{0+00,0}+w_{000}-1)I_{000}\geq 0\\
        &y_{0,0},y_{0,00},y_{00,0},y_{0,0+00},y_{0+00,0} \geq 0\\
        &w_0,w_{00},w_{000} \geq 0.
    \end{align*}
    which is equivalent to
    \begin{align*}
        \quad\text{ minimize:}\quad & w_0 + 2w_{00}+3w_{000}\\
        \text{subject to:}\quad &-2y_{0,0}-y_{0,00}-y_{00,0}-2y_{0,0+00}-2y_{0+00,0}+w_0\geq 0,\\
        &y_{0,0}-y_{0,00}-y_{00,0}+w_{00}\geq 1,\\
        &y_{0,00}+y_{00,0}+y_{0,0+00}+y_{0+00,0}+w_{000} \geq 1\\
        &y_{0,0},y_{0,00},y_{00,0},y_{0,0+00},y_{0+00,0} \geq 0\\
        &w_0,w_{00},w_{000} \geq 0.
    \end{align*}
    We can easily verify that setting all variables to zero except for $w_0=2, w_{00}=1$ and $y_{0+00,0}=1$ satisfies all the constraints. We explain in a moment how this feasible solution is obtained from the regular expression $R=(0+00)0$ of $L$. For now notice how the objective value of our solution is $4$ which happens to be the same as $|R|$.   
    \end{example}
    We now explain a procedure that recursively converts regular expressions to dual feasible solutions. Let $L$ be a language and let $R$ be a regular expression of $L$. The procedure outputs a feasible solution $(\hat{w},\hat{y})$ of $\mathcal{D}(L)$ with an objective value of $|R|$. If $R$ is a single term $s$, set $\hat{w}_s=1$ and set all the other variables to zero. Now suppose $R$ is not a single term. So it must be that either $R=(R_1+R_2)$ or $R=R_1R_2$ for some choice of regular expressions $R_1$ and $R_2$. Let $(\overline{w},\overline{y})$ and $(\underline{w},\underline{y})$ be feasible solutions obtained recursively for $R_1$ and $R_2$ respectively. 
    
    First suppose that $R=(R_1+R_2)$. Define a pair of vectors $$(\hat{w},\hat{y})\in {\real^+}^{\mathcal{C}_0(L)} \times {\real^+}^{\mathcal{C}_c(L)}$$ as follows. Let $\hat{w}_s=\overline{w}_s+\underline{w}_s$ and $\hat{y}_{K_1,K_2}=\overline{y}_{K_1,K_2}+\underline{y}_{K_1,K_2}$ for all $s\in\mathcal{C}_0(L)$ and $(K_1,K_2)\in \mathcal{C}_c(L)$. In these expressions interpret $\overline{w}_s,\underline{w}_s,\overline{y}_{K_1,K_2},\underline{y}_{K_1,K_2}$ as zeroes when they do not exist. So for example if $s\in \mathcal{C}_0(L)$ is such that $s\notin\mathcal{C}_0(L(R_1))$ we follow the convention that $\overline{w}_s=0$. Similarly if $(K_1,K_2)\in\mathcal{C}_c(L)$ is such that $(K_1,K_2)\notin\mathcal{C}_c(L(R_1))$ we assume $\underline{y}_{K_1,K_2}=0$. 
    
    Next suppose that $R=R_1R_2$. There could be more than one choice of $R_1$ and $R_2$. Choose any pair for which either $R_1$ does not end in a term or $R_2$ does not begin with a term (different choices here might yield different feasible solutions, but the objective values remain the same). This is always possible as $R$ is not a single term. Similar to the previous case define a pair of vectors $(\hat{w},\hat{y})$ by setting $\hat{w}_s=\overline{w}_s+\underline{w}_s$ and $\hat{y}_{K_1,K_2}=\overline{y}_{K_1,K_2}+\underline{y}_{K_1,K_2}$ for all $s\in\mathcal{C}_0(L)$ and $(K_1,K_2)\in \mathcal{C}_c(L)$. As before in these expressions interpret $\overline{w}_s,\underline{w}_s,\overline{y}_{K_1,K_2},\underline{y}_{K_1,K_2}$ as zeroes when they do not exist. Additionally set $\hat{y}_{L(R_1),L(R_2)}=1$. 
    
    Observe that the solution obtained from this procedure is an integer solution. The dual feasible solution given in Example \ref{ex:simple dual} provides one example of this procedure. We give one more example after the following theorem that proves the correctness of this procedure. 
    
    \begin{theorem}\label{thm:dual}
    Let $L$ be a language and let $R$ be a regular expression of $L$. The procedure discussed above yields a feasible solution of $\mathcal{D}(L)$ with an objective value of $|R|$. 
    \end{theorem}
    \begin{proof}
    First we prove that $\hat{w}_s$ is the number of times $s$ appears as a term in $R$ (i.e., the frequency of $s$ in $R$). We proceed by an induction on $|R|$. The case where $R$ is a single term is trivial (so the induction basis follows trivially as well). Suppose $R=(R_1+R_2)$. By induction hypothesis $\overline{w}_s$ and $\underline{w}_s$ are frequencies of $s$ in $R_1$ and $R_2$ respectively. Therefore $\hat{w}_s = \overline{w}_s+\underline{w}_s$ is the frequency of $s$ in $R$. Now suppose $R=R_1R_2$. We chose $R_1$ and $R_2$ so that terms in $R_1R_2$ are either terms in $R_1$ or terms in $R_2$. So the proof again follows from the induction hypothesis. 
    
    The length $|R|$ is the sum of lengths of terms of $R$. So from our claim above the objective value of $(\hat{w},\hat{y})$ is $J_L^T\hat{w}=\sum_{s\in \mathcal{C}_0(L)} |s|\hat{w}_s = |R|$.
    
    Next we show feasibility. Note that the positivity constraints of $\mathcal{D}(L)$ are trivially satisfied from the construction. We prove that all the other constraints are satisfied as well. The proof is again by an induction on $|R|$. Suppose $R$ is a single term $s$, i.e., $L=\{s\}$. Then $\hat{w}_s = 1$ and all the other variables are zero. So it holds that $\hat{w}= I_s = I_L$. 
    
    Now assume $R=(R_1+R_2)$. It holds that
    \begin{align*}
    \sum_{(K_1,K_2)\in\mathcal{C}_c(L)}\hat{y}_{K_1,K_2}(I_{K_1K_2} - I_{K_1} - I_{K_2}) + \hat{w}&= \sum_{(K_1,K_2)\in\mathcal{C}_c(L)}(\overline{y}_{K_1,K_2}+\underline{y}_{K_1,K_2})(I_{K_1K_2} - I_{K_1} - I_{K_2}) + \overline{w}+\underline{w}\\
    &\geq I_{L(R_1)} + I_{L(R_2)}\\
    &\geq I_{L(R)}.
    \end{align*}
    The first inequality follows from the induction hypothesis. The second inequality follows from $L(R)=L(R_1)+L(R_2)$.
    
    Next suppose $R=R_1R_2$. We can ensure that $R_1$ and $R_2$ are the same expressions as the ones the procedure chooses when the procedure is applied to $R$. Let $L_1=L(R_1)$ and $L_2=L(R_2)$. It is evident that due to our strict definition of regular expressions (not allowing empty language nor empty string), the expression $R_1R_2$ cannot be a subexpression of $R_1$ or $R_2$. So by construction $\overline{y}_{L_1,L_2}=\underline{y}_{L_1,L_2}=0$. So we can write
    \begin{align*}
    &\sum_{(K_1,K_2)\in\mathcal{C}_c(L)}\hat{y}_{K_1,K_2}(I_{K_1K_2} - I_{K_1} - I_{K_2}) + \hat{w}\\
    &\quad\quad=\sum_{\substack{(K_1,K_2)\in\mathcal{C}_c(L)\\(K_1,K_2)\neq (L_1,L_2)}}\hat{y}_{K_1,K_2}(I_{K_1K_2} - I_{K_1} - I_{K_2}) + \hat{y}_{L_1,L_2}(I_{L_1L_2} - I_{L_1} - I_{L_2}) + \hat{w}\\
    &\quad\quad= \sum_{(K_1,K_2)\in\mathcal{C}_c(L)}(\overline{y}_{K_1,K_2}+\underline{y}_{K_1,K_2})(I_{K_1K_2} - I_{K_1} - I_{K_2}) + I_{L_1L_2} - I_{L_1} - I_{L_2} + \overline{w}+\underline{w}\\
    &\quad\quad\geq I_{L_1} + I_{L_2} + I_{L_1L_2} - I_{L_1} - I_{L_2}\\
    &\quad\quad= I_{L_1L_2}.
    \end{align*}
    For the second equality we used the fact that $\hat{y}_{L_1,L_2}=1$ from the construction. The inequality follows from induction hypothesis.
    \end{proof}
    
    A simple corollary of this theorem is that $\opt(\mathcal{D}(L))\leq |R|$ for any regular expression $R$ of $L$. So together with the duality theorem of linear programming we obtained a second proof of Theorem \ref{thm:lowerbound}.
    
    \begin{example}\label{ex:all strings dual}
    Let $\Sigma$ be the binary alphabet, let $L=\Sigma^n$, and let $R_n=(0+1)^n$ be a regular expression of $L$. We discussed a primal feasible solution for this language in Example \ref{ex:all strings primal}. Here we provide a dual feasible solution. We do so by applying the procedure. Since $0$ and $1$ are the only terms appearing in $R_n$ the procedure yields $\hat{w}_0=\hat{w}_1=n$. Now suppose $n \geq 2$ and note the recursive relation $R_{n}=R_{n-1}R_1$. With this recursive relation in mind, the procedure yields $\hat{y}_{\Sigma^{k},\Sigma}=1$ for all $1\leq k\leq n-1$. The procedure yields zero for all other variables. 

    Let us directly verify that $(\hat{w},\hat{y})$ is a feasible solution. We must show that the following inequality holds
    \begin{align*}
    \sum_{(K_1,K_2)\in\mathcal{C}_c(\Sigma^n)}\hat{y}_{K_1,K_2}(I_{K_1K_2} - I_{K_1} - I_{K_2}) + \hat{w} \geq I_{\Sigma^n}.
    \end{align*}
    This follows from the following two observations. First it holds that
    \begin{align*}
    &\sum_{(K_1,K_2)\in\mathcal{C}_c(\Sigma^n)}\hat{y}_{K_1,K_2}(I_{K_1K_2} - I_{K_1} - I_{K_2})=
    \sum_{k=1}^{n-1}(I_{\Sigma^{k+1}} - I_{\Sigma^k} - I_{\Sigma})\\
    &= I_{\Sigma^{n}}-nI_{\Sigma}.
    \end{align*}
    Secondly we observe that $\hat{w}=n I_{\Sigma}$. 
    \end{example}
    
    We end this section by giving the dual of the strong formulation.   For $K\in\mathcal{C}(L)$ the vector $E_K\in\real^{\mathcal{C}(L)}$ is defined by
    \begin{align*}
            E_K(K')=\begin{cases} 
                          1 & K'=K, \\
                          0 & \text{otherwise},\\
                       \end{cases}
    \end{align*}
    for all $K'\in\mathcal{C}(L)$. Denote $E_{\{s\}}$ simply by $E_s$. The dual problem $\mathcal{D}_S(L)$ is as follows
    \begin{align*}
        \mathcal{D}_S(L)\quad\text{ minimize:}\quad &\sum_{s\in\mathcal{C}_0(L)} |s|W_s\\
        \text{subject to:}\quad &\sum_{(K_1,K_2)\in\mathcal{C}_c(L)}Y_{K_1,K_2}(E_{K_1K_2}-E_{K_1}-E_{K_2})\\ &\quad\quad+ \sum_{(K_1,K_2)\in\mathcal{C}_u(L)}Z_{K_1,K_2}(E_{K_1+K_2}-E_{K_1}-E_{K_2})+ \sum_{s\in\mathcal{C}_0(L)}W_sE_s \geq E_L,\\
        &Y_{K_1,K_2} \geq 0 \text{ for all }(K_1,K_2)\in\mathcal{C}_c(L)\\
        &Z_{K_1,K_2} \geq 0 \text{ for all }(K_1,K_2)\in\mathcal{C}_u(L)\\
        &W_s \geq 0 \text{ for all }{s\in \mathcal{C}_0(L)}.
    \end{align*}
    Note that $Y_{K_1,K_2},Z_{K_1,K_2},W_{s}$ are all real valued variables.
    
    The analogue of Theorem \ref{thm:dual} exists for the strong formulation in a straightforward manner. In particular the procedure for obtaining dual feasible solutions is very similar to the case of weak formulation. 

\section{Application of linear programming duality to lower bounds}\label{section:duality}
    A regular expression $R_{n,k}$ of $B(n,k)$ is defined in Ellul et al. \cite{ellul-regular-expressions} by means of a recursive relation as follows. Let $R_{n,0}=0^n$. For $k > \frac{n}{2}$ let $R_{n,k}=\tilde{R}_{n,n-k}$ where $\tilde{R}$ is the regular expression obtained from $R$ by flipping zeroes and ones. Finally when $k \leq \frac{n}{2}$, let
    \begin{align*}R_{n,k} = \Big(R_{\lfloor\frac{n}{2}\rfloor,0}R_{\lceil\frac{n}{2}\rceil,k} + R_{\lfloor\frac{n}{2}\rfloor,1}R_{\lceil\frac{n}{2}\rceil,k-1}+\cdots+R_{\lfloor\frac{n}{2}\rfloor,k}R_{\lceil\frac{n}{2}\rceil,0}\Big).
    \end{align*}
    It is also stated in \cite{ellul-regular-expressions} that $|R_{n,k}|=O(n\log^k(n))$ where the hidden constant depends on $k$. It is also mentioned there that $R_{n,1}$ is optimal up to a constant multiplicative factor, a proof of which is given in this note in Example \ref{ex:binomial}. In this section we take the first major step in proving that $R_{n,k}$ is asymptotically optimal for all values of $k$. 

    Consider the binomial language $B(n,k)$ and the associated primal problem
    \begin{align*}
        \mathcal{P}(B(n,k))\quad\text{ maximize:}\quad &\sum_{s\in B(n,k)}x_s\\
        \text{subject to:}\quad & \sum_{s\in K_1K_2}x_s \leq \sum_{s\in K_1}x_s + \sum_{s\in K_2}x_s  \quad\text{ for all } (K_1,K_2)\in \mathcal{C}_c(B(n,k)),\\
        &0 \leq x_s \leq |s|\quad\text{ for all } s \in \mathcal{C}_0(B(n,k)).
    \end{align*}
    The number of variables in the primal is $|\mathcal{C}_0(B(n,k))|=\sum_{l=0}^{k}\sum_{m=l}^n \binom{m}{l}$ which is a polynomial in $n$. On the other hand the number of constraints is exponential in $n$ from the observation that $|\mathcal{C}(B(n,k))|=\sum_{l=0}^k\sum_{m=l}^n 2^{\binom{m}{l}}$.
    
    A strategy for proving lower bounds on $B(n,k)$ is to guess an $x$ then check if it is feasible in $\mathcal{P}(B(n,k))$. However this is not an easy task as seen from the number of constraints that must be verified. In this section we see a relaxation $\mathcal{P}'(n,k)$ to $\mathcal{P}(B(n,k))$ with only polynomially many constraints. Then using duality theorem we show that $\opt(\mathcal{P}'(n,k))\leq |R|$ where $R$ is any regular expression of $B(n,k)$. Therefore any feasible solution to this tractable relaxation proves a lower bound on $B(n,k)$.
    
    Let $\mathcal{P}'(n,k)$ be defined as follows
    \begin{align*}
        \mathcal{P}'(n,k)\quad\text{ maximize:}\quad &\sum_{s\in B(n,k)}x_s\\
        \text{subject to:}\quad & \sum_{s\in B(n_1,k_2)B(n_2,k_2)}x_s \leq \sum_{s\in B(n_1,k_1)}x_s + \sum_{s\in B(n_2,k_2)}x_s  \quad\text{ for all } (n_1,k_1,n_2,k_2)\in \mathcal{C}_c(n,k),\\
        &0 \leq x_s \leq |s|\quad\text{ for all } s \in \mathcal{C}_0(B(n,k)).
    \end{align*}
    This is clearly a relaxation of $\mathcal{P}(B(n,k))$ because the constraints $$\sum_{s\in B(n_1,k_2)B(n_2,k_2)}x_s \leq \sum_{s\in B(n_1,k_1)}x_s + \sum_{s\in B(n_2,k_2)}x_s$$ all appear in $\mathcal{P}(B(n,k))$. Therefore it holds that $\opt(\mathcal{P}(B(n,k)))\leq\opt(\mathcal{P}'(n,k))$.
    
    There is a connection between this relaxation and the class of \emph{complete regular expressions} which we define next. Let $R$ be a regular expression for $B(n,k)$ and let $R_1R_2$ be an arbitrary subexpression of $R$. Due to the special structure of $\mathcal{C}_c(B(n,k))$ it must hold that $L(R_1)$ and $L(R_2)$ are subsets of some binomial languages. Inspired by this observation a regular expression for $B(n,k)$ is a \emph{complete regular expression} if for every subexpression $R_1R_2$ it holds that both $L(R_1)$ and $L(R_2)$ are binomial languages, i.e., $L(R_1)=B(n_1,k_1), L(R_2)=B(n_2,k_2)$ for some values of $n_1,k_1,n_2,k_2$. For example the regular expression $R_{n,k}$ is a complete regular expression for $B(n,k)$. Now similar to the proof of Theorem \ref{thm:lowerbound}, we can easily argue that $\opt(\mathcal{P}'(n,k))$ is a lower bound on the length of any complete regular expression of $B(n,k)$. We will see however that a stronger statement holds. The corollary to our next theorem states that $\opt(\mathcal{P}'(n,k))$ is a lower bound on the length of any regular expression of $B(n,k)$ and not just complete regular expressions. 
    
    To proceed we must first observe that the dual of $\mathcal{P}'(n,k)$ is as follows
    \begin{align*}
        \mathcal{D}'(n,k)\quad\text{ minimize:}\quad &J_{B(n,k)}^Tw\\
        \text{subject to:}\quad &\sum_{(n_1,k_1,n_2,k_2)\in\mathcal{C}_c(n,k)}y_{n_1,k_1,n_2,k_2}(I_{B(n_1,k_1)B(n_2,k_2)} - I_{B(n_1,k_1)} - I_{B(n_2,k_2)}) + w \geq I_{B_{n,k}},\\
        &y_{n_1,k_1,n_2,k_2} \geq 0 \quad\text{ for all }(n_1,k_1,n_2,k_2)\in\mathcal{C}_c(n,k),\\
        &w \in {\real^+}^{\mathcal{C}_0(B(n,k))}.
    \end{align*}
    Note that $(w,y)\in {\real^+}^{\mathcal{C}_0(B(n,k))} \times {\real^+}^{\mathcal{C}_c(n,k)}.$
    We are now ready to state and prove the main theorem of this note. The proof is very similar to the proof of Theorem \ref{thm:dual}. 
    \begin{theorem}\label{thm:main}
    Let $B(n,k)$ be a binomial language. There exit a procedure that turns any regular expression of $B(n,k)$ into a feasible solution of $\mathcal{D}'(n,k)$ with an objective value that is equal to the length of the regular expression.  
    \end{theorem}
    \begin{proof}
    By the definition of closure,if $R$ is subexpressions of any regular expression of $B(n,k)$, then $L(R)\in \mathcal{C}(B(n,k))$. First we give a more general procedure that converts any regular expression $R$ for which $L(R)\in \mathcal{C}(B(n,k))$ to a pair of vectors in ${\real^+}^{\mathcal{C}_0(B(n,k))} \times {\real^+}^{\mathcal{C}_c(n,k)}$.
    
    Let $R$ be a regular expression for a language in $\mathcal{C}(B(n,k))$. On input $R$, the procedure outputs a pair of vectors $(\hat{w},\hat{y})\in {\real^+}^{\mathcal{C}_0(B(n,k))} \times {\real^+}^{\mathcal{C}_c(n,k)}$ (not necessarily feasible) as follows. First suppose $R$ is a single term $t$ and $t\in B(m,l)$ for some $m,l$. Set $\hat{w}_s=\frac{1}{|B(m,l)|}$ for all $s\in B(m,l)$ and set all the other variables to zero. 
    
    Now suppose $R$ is not a single term. So it must be that either $R=(R_1+R_2)$ or $R=R_1R_2$ for some choice of regular expressions $R_1$ and $R_2$. Let $(\overline{w},\overline{y})$ and $(\underline{w},\underline{y})$ be the outputs of the procedure on $R_1$ and $R_2$ respectively. 
    
    First suppose that $R=(R_1+R_2)$. Define a pair of vectors $(\hat{w},\hat{y})$ as follows. Let $\hat{w}_s=\overline{w}_s+\underline{w}_s$ and $\hat{y}_{n_1,k_1,n_2,k_2}=\overline{y}_{n_1,k_1,n_2,k_2}+\underline{y}_{n_1,k_1,n_2,k_2}$ for all $s\in\mathcal{C}_0(B(n,k))$ and $(n_1,k_1,n_2,k_2)\in \mathcal{C}_c(n,k)$. 
    
    Next suppose that $R=R_1R_2$. There could be more than one choice of $R_1$ and $R_2$. Choose any pair for which either $R_1$ does not end in a term or $R_2$ does not begin with a term. There must exists $(m_1,l_1,m_2,l_2)\in \mathcal{C}_c(n,k)$ such that $L(R_1)\subseteq B(m_1,l_1)$ and $L(R_2)\subseteq B(m_2,l_2)$. Consequently $L(R)\subseteq B(m_1+m_2,l_1+l_2)$. Similar to the previous case define a pair of vectors $(\hat{w},\hat{y})$ by setting $\hat{w}_s=\overline{w}_s+\underline{w}_s$ and $\hat{y}_{n_1,k_1,n_2,k_2}=\overline{y}_{n_1,k_1,n_2,k_2}+\underline{y}_{n_1,k_1,n_2,k_2}$ for all $s\in\mathcal{C}_0(B(n,k))$ and $(n_1,k_1,n_2,k_2)\in \mathcal{C}_c(n,k)$. Additionally set $$\hat{y}_{m_1,l_1,m_2,l_2}=\frac{|L(R_1)||L(R_2)|}{|B(m_1+m_2,l_1+l_2)|}=\frac{|L(R)|}{|B(m_1+m_2,l_1+l_2)|}.$$ 
    
    First we prove that the objective value is $$J_{B(n,k)}^T\hat{w}=\sum_{s\in \mathcal{C}_0(B(n,k))}|s|\hat{w}_s = |R|.$$ We proceed by an induction on $|R|$. Suppose $R$ is a single term $t$ and $t\in B(m,l)$ for some $m,l$. From the construction it holds that $\hat{w}_s = \frac{1}{|B(m,l)|}$ for all $s\in B(m,l)$ and it is zero everywhere else. Thus $\sum_{s\in \mathcal{C}_0(B(n,k))}|s|\hat{w}_s=m=|R|$. Now suppose $R=(R_1+R_2)$. By the induction hypothesis $\sum_{s\in \mathcal{C}_0(B(n,k))}|s|\overline{w}_s = |R_1|$ and $\sum_{s\in \mathcal{C}_0(B(n,k))}|s|\underline{w}_s = |R_2|$. Now since $\hat{w}_s = \overline{w}_s+\underline{w}_s$ it holds that $\sum_{s\in \mathcal{C}_0(B(n,k))}|s|\hat{w}_s = |R_1|+|R_2|=|R_1+R_2|$. The case of $R=R_1R_2$ is similar. 
    
    Now suppose $L(R)\subseteq B(m,l)$ for some $m,l$. We prove the following $$\sum_{(n_1,k_1,n_2,k_2)\in\mathcal{C}_c(n,k)}\hat{y}_{n_1,k_1,n_2,k_2}(I_{B(n_1,k_1)B(n_2,k_2)} - I_{B(n_1,k_1)} - I_{B(n_2,k_2)}) + \hat{w} \geq \frac{|L(R)|}{|B(m,l)|}I_{B(m,l)},$$
    by an induction on $|R|$.
    
    Suppose $R$ is a single term $s$. Then from the construction it holds that
    \begin{align*}
    \sum_{(n_1,k_1,n_2,k_2)\in\mathcal{C}_c(n,k)}\hat{y}_{n_1,k_1,n_2,k_2}(I_{B(n_1,k_1)B(n_2,k_2)} - I_{B(n_1,k_1)} - I_{B(n_2,k_2)}) + \hat{w} &= \hat{w}\\
    &= \frac{1}{|B(m,l)|}I_{B(m,l)}.
    \end{align*}
    
    Now suppose $R=(R_1+R_2)$. It necessarily must hold that $R_1,R_2\subseteq B(m,l)$. Now we can write
    \begin{align*}
    &\sum_{(n_1,k_1,n_2,k_2)\in\mathcal{C}_c(n,k)}\hat{y}_{n_1,k_1,n_2,k_2}(I_{B(n_1,k_1)B(n_2,k_2)} - I_{B(n_1,k_1)} - I_{B(n_2,k_2)}) + \hat{w}\\&\quad\quad=\sum_{(n_1,k_1,n_2,k_2)\in\mathcal{C}_c(n,k)}(\overline{y}_{n_1,k_1,n_2,k_2}+\underline{y}_{n_1,k_1,n_2,k_2})(I_{B(n_1,k_1)B(n_2,k_2)} - I_{B(n_1,k_1)} - I_{B(n_2,k_2)}) + \overline{w}+\underline{w}\\
    &\quad\quad\geq \frac{|L(R_1)|}{|B(m,l)|}I_{B(m,l)} + \frac{|L(R_2)|}{|B(m,l)|}I_{B(m,l)}\\
    &\quad\quad\geq \frac{|L(R)|}{|B(m,l)|} I_{B(m,l)}.
    \end{align*}
    The first inequality follows from the induction hypothesis. The second inequality follows from $|L(R_1)|+|L(R_2)|\geq |L(R)|$
    
    Next suppose that $R=R_1R_2$. We can ensure that $R_1$ and $R_2$ are the same expressions as the ones the procedure chooses when the procedure is applied to $R$. Let $L=L(R), L_1=L(R_1)$ and $L_2=L(R_2)$. There must exist $m_1,l_1,m_2,l_2$ such that $L_1\subseteq B(m_1,l_1), L_2\subseteq B(m_2,l_2)$ and $L\subseteq B(m_1+m_2,l_1+l_2)$ (i.e, $m$ and $l$ from before are $m_1+m_2$ and $l_1+l_2$ respectively). Now we write
    \begin{align}
    &\sum_{(n_1,k_1,n_2,k_2)\in\mathcal{C}_c(n,k)}\hat{y}_{n_1,k_1,n_2,k_2}(I_{B(n_1,k_1)B(n_2,k_2)} - I_{B(n_1,k_1)} - I_{B(n_2,k_2)}) + \hat{w}\nonumber\\
    &\quad\quad=\sum_{\substack{(n_1,k_1,n_2,k_2)\in\mathcal{C}_c(n,k)\nonumber\\(n_1,k_1,n_2,k_2)\neq (m_1,l_1,m_2,l_2)}}\hat{y}_{n_1,k_1,n_2,k_2}(I_{B(n_1,k_1)B(n_2,k_2)} - I_{B(n_1,k_1)} - I_{B(n_2,k_2)})\nonumber\\&\quad\quad\quad\quad\quad\quad\quad\quad\quad + \hat{y}_{m_1,l_1,m_2,l_2}(I_{B(m_1,l_1)B(m_2,l_2)} - I_{B(m_1,l_1)} - I_{B(m_2,l_2)}) + \hat{w}\nonumber\\
    &\quad\quad=\sum_{(n_1,k_1,n_2,k_2)\in\mathcal{C}_c(n,k)}(\overline{y}_{n_1,k_1,n_2,k_2}+\underline{y}_{n_1,k_1,n_2,k_2})(I_{B(n_1,k_1)B(n_2,k_2)} - I_{B(n_1,k_1)} - I_{B(n_2,k_2)})\nonumber\\&\quad\quad\quad\quad\quad\quad\quad\quad\quad + \frac{|L_1||L_2|}{|B(m_1+m_2,l_1+l_2)|}(I_{B(m_1,l_1)B(m_2,l_2)} - I_{B(m_1,l_1)} - I_{B(m_2,l_2)}) + \overline{w}+\underline{w},\label{eq:main theorem:induction hypothesis}
    \end{align}
    in which we used the facts that $\hat{y}_{n_1,k_1,n_2,k_2} = \overline{y}_{n_1,k_1,n_2,k_2}+\underline{y}_{n_1,k_1,n_2,k_2}$ when $(n_1,k_1,n_2,k_2)\neq (m_1,l_1,m_2,l_2)$, $\hat{w}=\overline{w}+\underline{w}$ and $\hat{y}_{m_1,l_1,m_2,l_2}=\frac{|L_1||L_2|}{|B(m_1+m_2,l_1+l_2)|}$. From the induction hypothesis the followings hold
    \begin{align*}
    \sum_{(n_1,k_1,n_2,k_2)\in\mathcal{C}_c(n,k)}\overline{y}_{n_1,k_1,n_2,k_2}(I_{B(n_1,k_1)B(n_2,k_2)} - I_{B(n_1,k_1)} - I_{B(n_2,k_2)})+\overline{w} &\geq \frac{|L_1|}{|B(m_1,k_1)|}I_{m_1,l_1},\\
    \sum_{(n_1,k_1,n_2,k_2)\in\mathcal{C}_c(n,k)}\underline{y}_{n_1,k_1,n_2,k_2}(I_{B(n_1,k_1)B(n_2,k_2)} - I_{B(n_1,k_1)} - I_{B(n_2,k_2)})+\underline{w} &\geq \frac{|L_2|}{|B(m_2,k_2)|}I_{m_2,l_2}.
    \end{align*}
    So we can continue (\ref{eq:main theorem:induction hypothesis}) as follows
    \begin{align*}
    &\sum_{(n_1,k_1,n_2,k_2)\in\mathcal{C}_c(n,k)}\hat{y}_{n_1,k_1,n_2,k_2}(I_{B(n_1,k_1)B(n_2,k_2)} - I_{B(n_1,k_1)} - I_{B(n_2,k_2)}) + \hat{w}\\
    &\quad\quad\geq \frac{|L_1|}{|B(m_1,l_1)|}I_{B(m_1,l_1)} + \frac{|L_2|}{|B(m_2,l_2)|}I_{B(m_2,l_2)} + \frac{|L_1||L_2|}{|B(m_1+m_2,l_1+l_2)|}(I_{B(m_1,l_1)B(m_2,l_2)} - I_{B(m_1,l_1)} - I_{B(m_2,l_2)})\\
    &\quad\quad\geq \frac{|L_1||L_2|}{|B(m_1+m_2,l_1+l_2)|}I_{B(m_1,l_1)B(m_2,l_2)}\\
    &\quad\quad= \frac{|L|}{|B(m_1+m_2,l_1+l_2)|}I_{B(m_1,l_1)B(m_2,l_2)},
    \end{align*}
    in which we used the following inequalities
    \begin{align*}
    \frac{|L_1|}{|B(m_1,l_1)|} \geq  \frac{|L_1||B(m_2,l_2)|}{|B(m_1+m_2,l_1+l_2)|}\geq \frac{|L_1||L_2|}{|B(m_1+m_2,l_1+l_2)|},\\
    \frac{|L_2|}{|B(m_2,l_2)|} \geq \frac{|B(m_1,l_1)||L_2|}{|B(m_1+m_2,l_1+l_2)|} \geq \frac{|L_1||L_2|}{|B(m_1+m_2,l_1+l_2)|}.
    \end{align*}
    These inequalities can be obtained from the following observation
    \begin{align*}
    |B(m_1+m_2,l_1+l_2)|=\binom{m_1+m_2}{l_1+l_2}\geq \binom{m_1}{l_1}\binom{m_2}{l_2}=|B(m_1,l_1)||B(m_2,l_2)|.
    \end{align*}
    Now suppose $R$ is a regular expression of $B(n,k)$. We proved that $(\hat{w},\hat{y})$ obtained from the procedure satisfies
    \begin{align*}
    \sum_{(n_1,k_1,n_2,k_2)\in\mathcal{C}_c(n,k)}\hat{y}_{n_1,k_1,n_2,k_2}(I_{B(n_1,k_1)B(n_2,k_2)} - I_{B(n_1,k_1)} - I_{B(n_2,k_2)}) + \hat{w} \geq I_{B(n,k)}.
    \end{align*}
    So all the constraints in $\mathcal{P}'(n,k)$ are satisfied and $(\hat{w},\hat{y})$ is a feasible solution. Furthermore we showed that the objective value is $J_{B(n,k)}^T\hat{w} = |R|$. This completes the proof.
    \end{proof}
    \begin{corollary} Let $R$ be any regular expression of $B(n,k)$. It holds that $\opt(\mathcal{P}'(n,k))=\opt(\mathcal{D}'(n,k))\leq |R|$.
    \end{corollary}
    \begin{proof}
    From Theorem \ref{thm:main} it holds that $\opt(\mathcal{D}'(n,k))\leq |R|$, and from the duality theorem it holds that $\opt(\mathcal{P}'(n,k))=\opt(\mathcal{D}'(n,k))$.
    \end{proof}
    
    The problem $\mathcal{P}'(n,k)$ is computationally tractable, i.e., its size is a polynomial in $n$. So this problem can be solved efficiently. We numerically verified that $\opt(\mathcal{P}'(n,k)) = |R_{n,k}|$ for many examples of $n$ and 
    $k$. So there is good evidence that the answer to the following open problem is affirmative.
    \begin{problem} Is it true that $\opt(\mathcal{P}'(n,k)) = |R_{n,k}|$?
    \end{problem}
    We suspect that the following holds true as well.
    \begin{problem}\label{problem:collapse} Is the following true? $$\opt(\mathcal{P}(B(n,k))) = \opt(\mathcal{P}'(n,k)) = |R_{n,k}| = \opt(\mathcal{D}'(n,k))=\opt(\mathcal{D}(B(n,k))).$$
    A simpler problem is perhaps to show that
    $$\opt(\mathcal{P}(B(n,k)))=\opt(\mathcal{D}(B(n,k)))= \opt(\mathcal{P}'(n,k)) = \opt(\mathcal{D}'(n,k)).$$
    We know that 
    $$\opt(\mathcal{P}(B(n,k)))=\opt(\mathcal{D}(B(n,k))) \leq \opt(\mathcal{P}'(n,k)) = \opt(\mathcal{D}'(n,k)),$$
    so a strategy for proving that they are all equal is to prove the existence of a procedure that converts feasible solutions of $\mathcal{D}(B(n,k))$ to feasible solutions of $\mathcal{D}'(n,k)$ with the same objective value.
    \end{problem}
    
    In the next section we give a feasible solution of $\mathcal{P}'(n,k)$ with an objective value of $\Theta(n\log^k(n))$. This will complete the proof that $R_{n,k}$ is asymptotically optimal.

\section{A matching lower bound for binomial languages}\label{section:feasible}
    The goal of this section is to construct a feasible solution of $\mathcal{P}'(n,k)$ with an objective value of $\Theta(n\log^k(n))$. Such a construction completes the proof that $R_{n,k}$ is asymptotically optimal.

    The alphabet $\Sigma$ in this section is the binary alphabet. We say a map $g:\Sigma^+\rightarrow \mathcal{R}^+$ is a feasible solution of $\mathcal{P}'(n,k)$ if $\overline{x}$ defined such that $\overline{x}_s=g(s)$ for all $s\in\mathcal{C}_0(n,k)$ is a feasible solution of $\mathcal{P}'(n,k)$. Similarly the objective value of $g$ in $\mathcal{P}'(n,k)$ refers to the quantity $\sum_{s\in B(n,k)} g(s)$. 
    
    We begin by defining a map $g:\Sigma^+\rightarrow \mathcal{R}^+$. First we specify the action of this map on strings in $B(n,0)\cup B(n,1)$. Let $g(0^n)=n$ for all $n > 0$. Let $g(s)=\ln(en)$ for all $s\in B(n,1)$ and for all $n > 0$. From Example \ref{ex:binomial} and the discussion in the previous section we know that $g$ is feasible in $\mathcal{P}'(n,1)$ and its objective value is $n\ln(en)$.
    
    Next we give the full specifications of $g$. For any string $s$ for which $k=|s|_1 > 1$ let 
    \begin{align*}
        g(s) &= \alpha_{k-1}\frac{\ln^{k-1}(p)}{p^{k-1}}.
    \end{align*}
    where $p$ is the maximum number of symbols between any pairs of $1$s in $s$, and  $\alpha_1,\alpha_2,\ldots$ are real positive constants we choose later. For example we have
    \begin{align*}
    g({11}) &= \frac{\alpha_1\ln(2)}{2},g({010001})=\frac{\alpha_1\ln(5)}{5}, g({010101})=\frac{\alpha_2\ln^2(5)}{5^2},g({01010010100})=\frac{\alpha_3\ln^3(8)}{8^3}.
    \end{align*}
    
    Before we can calculate the objective value of $g$ in $\mathcal{P}'(n,k)$ for $k > 1$ we need to review a few identities from calculus.
    
    Let $k\geq 0$ be a constant. The followings hold
    \begin{align*}
    \binom{n}{k} &= \Theta(n^k),\\
    \sum_{i=1}^{n} \frac{\ln^k(i)}{i} &= \Theta\Big(\ln^{k+1}(n)\Big),\\
    \sum_{i=1}^{n} \ln^k(i) &= \Theta\Bigl(n\ln^k(n)\Bigr).    
    \end{align*}
    
    When we write variable length sums involving $\Theta$ of the same function, e.g, $\sum_{i=m+1}^n \Theta(h(i))$ it is always assumed that the implied constants are the same in all the terms involving $\Theta$. 
    We are now ready to show that the objective value of $g$ in $\mathcal{P}'(n,k)$ is $\Theta(n\ln^k(n))$. We know this holds for $k=1$. For $k > 1$ this is calculated as follows
    \begin{align*}
    \sum_{s\in B(n,k)} g(s)&= \sum_{j=0}^{n-k}\sum_{i=k}^{n-j}\sum_{s\in B(i-2,k-2)} g(0^j1s10^{n-i-j})\\
    &= \sum_{j=0}^{n-k}\sum_{i=k}^{n-j}\sum_{s\in B(i-2,k-2)} \alpha_{k-1}\frac{\ln^{k-1}(i)}{i^{k-1}}\\
    &= \sum_{j=0}^{n-k}\sum_{i=k}^{n-j}\alpha_{k-1}\binom{i-2}{k-2} \frac{\ln^{k-1}(i)}{i^{k-1}}\\
    &= \sum_{j=0}^{n-k}\sum_{i=k}^{n-j} \Theta\Big(\frac{\ln^{k-1}(i)}{i}\Big)\\
    &= \sum_{j=0}^{n-k} \Theta\Big(\ln^{k}(n-j)\Big)\\
    &= \Theta\Big(n\ln^{k}(n)\Big).
    \end{align*}
        
    Next we prove the feasibility of $g$ in $\mathcal{P}'(n,k)$.
    For the special case of $k=2$ we can in fact prove that $g$ is feasible not only in $\mathcal{P}'(n,k)$ but in $\mathcal{P}(B(n,k))$ as well. A proof of this fact is given in the next theorem.
    
    \begin{theorem}\label{thm:k=2} The constant $\alpha_1$ can be chosen so that the map $g$ is feasible in $\mathcal{P}(B(n,2))$.
    \end{theorem}
    \begin{proof}
    Let $(K_1,K_2)\in \mathcal{C}_c(B(n,2))$. We need to show that \begin{align}\label{eq:thm:k=2:eq1}
    \sum_{s\in K_1K_2}g(s) \leq \sum_{s\in K_1} g(s) + \sum_{s\in K_2} g(s).
    \end{align}
    First note that $K_1\subseteq B(n_1,k_1)$ and $K_2\subseteq B(n_2,k_2)$ for some $(n_1,k_1,n_2,k_2)\in \mathcal{C}_c(n,2)$. The only possibilities for $(k_1,k_2)$ are $(0,0),(0,1),(1,0),(1,1),(0,2),(2,0)$. The inequality (\ref{eq:thm:k=2:eq1}) is nontrivial only when $(k_1,k_2)=(1,1)$. So suppose that $K_1\subseteq B(n_1,1)$ and $K_2\subseteq B(n_2,1)$. With this assumption the right hand side of (\ref{eq:thm:k=2:eq1}) is simply $|K_1|\ln(en_1) + |K_2|\ln(en_2)$. So we need to show that
    \begin{align}\label{eq:thm:k=2:eq2}
    \sum_{s\in K_1K_2}g(s) \leq |K_1|\ln(en_1) + |K_2|\ln(en_2).
    \end{align}
    Now suppose $$K_1 = \left\{0^{n_1-p_1-1}10^{p_1},\ldots,0^{n_1-p_{|K_1|}-1}10^{p_{|K_1|}}\right\}$$ for distinct nonnegative integers $p_1,\ldots,p_{|K_1|}$. Also suppose $$K_2 = \left\{0^{q_1}10^{n_2-q_1-1},\ldots,0^{q_{|K_2|}}10^{n_2-q_{|K_2|}-1}\right\}$$  for distinct nonegative integers $q_1,\ldots, q_{|K_2|}$. The inequality (\ref{eq:thm:k=2:eq2}) is then equivalent to the following
    \begin{align*}
        \sum_{i=1}^{|K_1|}\sum_{j=1}^{|K_2|}\alpha_1\frac{\ln(p_i + q_j+2)}{p_i+q_j+2} \leq |K_1|\ln(en_1) + |K_2|\ln(en_2).
    \end{align*}
    The right hand side does not depend on $p_i$s and $q_j$s. On the other hand, since $\frac{\ln(x)}{x}$ is a decreasing function when $x \geq 3$, the left hand side is maximized\footnote{A more precise analysis must take into account that $\frac{\ln(2)}{2}< \frac{\ln(3)}{3}$. This can be handled by means of a simple case analysis.} when $p_i=i-1$ and $q_j=j-1$ for $i\in\{1,\ldots,|K_1|\}$ and $j\in\{1,\ldots,|K_2|\}$. So we only need to show that 
    \begin{align*}
        \sum_{i=1}^{|K_1|}\sum_{j=1}^{|K_2|}\alpha_1\frac{\ln(i +j)}{i+j} \leq |K_1|\ln(en_1) + |K_2|\ln(en_2).
    \end{align*}
    This effectively means that we only need to show that the inequality (\ref{eq:thm:k=2:eq1}) holds for pairs of languages $K_1=B(n_1,1)$ and $K_2=B(n_2,1)$ and for all choices of $n_1$ and $n_2$. In other words feasibility of $g$ in $\mathcal{P}'(n,2)$ implies feasibility in $\mathcal{P}(B(n,2))$. The feasibility in $\mathcal{P}'(n,2)$ is proved in the next lemma.
    \end{proof}
    
    In the beginning of this section we stated the asymptotic relations $\sum_{i=1}^{n} \frac{\ln(i)}{i} = O\Big(\ln^2(n)\Big)$ and $\sum_{i=1}^{n} \ln^2(i) = O\Bigl(n\ln^2(n)\Bigr)$. We state without proof that these can be strengthened in the sense that follows. There exists a constant $\beta > 0$ such that 
    \begin{align}\label{eq:strict notion 1}
    \sum_{i=m+1}^{n} \frac{\ln(i)}{i} &\leq \beta (\ln^2(n)-\ln^2(m))
    \end{align}
    for all $n\geq m > 0$. Similarly, there exists a constant $\gamma> 0$ such that
    \begin{align}\label{eq:strict notion 3}
    \sum_{i=m+1}^{n} \ln^2(i)\leq \gamma \Bigl(n\big(\ln^2(n)-2\ln(n)+2\big)-m\big(\ln^2(m)-2\ln(m)+2\big)\Bigr),
    \end{align}
    for all $n\geq m > 0$. In fact we mentioned (\ref{eq:strict notion 3}) just to provide intuition for the following fact which we use in the proof of the following lemma. There exists a constant $\gamma'$ such that
    \begin{align}\label{eq:strict notion 2}
    &\sum_{i=1}^{m} (\ln^2(i+n)-\ln^2(i))\nonumber\\ &\quad\leq \gamma' \Bigl((m+n)\big(\ln^2(m+n)-2\ln(m+n)\big)-m\big(\ln^2(m)-2\ln(m)\big)-n\big(\ln^2(n)-2\ln(n)\big)\Bigr),
    \end{align}
    for all $m,n > 0$.
    \begin{lemma}\label{lemma:the first complete sign} The constant $\alpha_1$ can be chosen such that the following holds
    \begin{align*} &\sum_{s\in B(m,1)B(n,1)} g(s) \leq \sum_{s\in B(m,1)} g(s) + \sum_{s\in B(n,1)} g(s),
    \end{align*}
    for all $m,n\geq 1$.
    \end{lemma}
    
    \begin{proof}
    We must show that there exists $\alpha_1$ such that
    \begin{align*}
        \sum_{i=1}^{m}\sum_{j=1}^{n}\alpha_1\frac{\ln(i +j)}{i+j} \leq m\ln(em) + n\ln(en),
    \end{align*}
    for all $m,n \geq 1$.  First from (\ref{eq:strict notion 1}) and (\ref{eq:strict notion 2}) it holds that
    \begin{align}
        &\sum_{i=1}^{m}\sum_{j=1}^{n}\frac{\ln(i+j)}{i+j} \nonumber\\&\quad\leq \sum_{i=1}^{m} \beta(\ln^2(i+n)-\ln^2(i))\nonumber\\
        &\quad\leq \beta\gamma' \Bigl((m+n)\big(\ln^2(m+n)-2\ln(m+n)\big)-m\big(\ln^2(m)-2\ln(m)\big)-n\big(\ln^2(n)-2\ln(n)\big)\Bigr).\label{eq:lemma:the first complete sign:eq1}
    \end{align}
    Using the inequality $$(m+n)\ln^2(m+n)-m\ln^2(m)-n\ln^2(n)\leq 2(m+n)\ln(m+n)$$ that is proved in the corollary to the next lemma we continue (\ref{eq:lemma:the first complete sign:eq1}) as follows
    \begin{align*}
        &\sum_{i=1}^{m}\sum_{j=1}^{n}\frac{\ln(i+j)}{i+j}\\
        &\quad\leq \beta\gamma' \Bigl((m+n)\big(\ln^2(m+n)-2\ln(m+n)\big)-m\big(\ln^2(m)-2\ln(m)\big)-n\big(\ln^2(n)-2\ln(n)\big)\Bigr)\\
        &\quad\leq \beta\gamma'\big(2m\ln(m) + 2n\ln(n)\big).
    \end{align*}
    So choosing $\alpha_1=\frac{1}{2\beta\gamma'} > 0$ completes the proof.
    \end{proof}
    
    \begin{lemma}\label{lemma:convexity} For any integer $k \geq 0$, the function $h_k(\alpha)=\alpha\ln^k(\alpha)$ is convex over $[1,\infty)$.
    \end{lemma}
    \begin{proof}
    For $k=0$ the function is linear so it is convex. Now assume $k\geq 1$. Proof follows from the positivity of the second derivative which is calculated as follows
    \begin{align*}
    h_k'(\alpha) &= \ln^k(\alpha) + k\ln^{k-1}(\alpha)\\
    h_k''(\alpha) &= k\frac{\ln^{k-1}(\alpha)}{\alpha} + k(k-1)\frac{\ln^{k-2}(\alpha)}{\alpha}.
    \end{align*}
    \end{proof}
    
    \begin{corollary}\label{lemma:convexity consequence} For all $m,n\geq 1$, it holds that \begin{align*}
    &(m+n)\ln^2(m+n) - m\ln^2(m) - n\ln^2(n)\leq 2(m+n)\ln(m+n).
    \end{align*}
    \end{corollary}
    \begin{proof}
    By convexity of $h_2(\alpha)=\alpha\ln^2(\alpha)$ it holds that $$\frac{m+n}{2}\ln^2(\frac{m+n}{2})\leq \frac{m}{2}\ln^2(m) + \frac{n}{2}\ln^2(n)$$ or equivalently $$(m+n)\ln^2(\frac{m+n}{2})\leq m\ln^2(m) + n\ln^2(n)$$ The proof follows when we expand the left hand side as follows
    \begin{align*}
    (m+n)\ln^2(\frac{m+n}{2}) &= (m+n)(\ln(m+n)-\ln(2))^2\\
    &= (m+n)\ln^2(m+n) - 2\ln(2)(m+n)\ln(m+n)+(m+n)\ln^2(2).
    \end{align*}
    Observe that $2\ln(2)<2$.
    \end{proof}
    
    Unlike $k=2$, we do not know if the feasibility of $g$ in $\mathcal{P}'(n,k)$ implies the feasibility of $g$ in $\mathcal{P}(B(n,k))$ for $k\geq 3$. 
    
    Recall the identity
    \begin{align*}
    \sum_{i=1}^{n} \frac{\ln^k(i)}{i} &= \Theta\Big(\ln^{k+1}(n)\Big).
    \end{align*}
    Similar to (\ref{eq:strict notion 1})-(\ref{eq:strict notion 2}) a stronger statement holds here as well. There exist positive constants $\beta_1$ and $\beta_2$ (depending on $k$) such that $$\beta_1 \ln^{k+1}(n) \leq \sum_{i=1}^{n} \frac{\ln^k(i)}{i} \leq \beta_2\ln^{k+1}(n)$$ for all $n \geq 1$ (and not just for sufficiently large $n$). A similar situation holds for the identity 
    \begin{align*}
    \binom{n}{k} &= \Theta(n^k).
    \end{align*}
    Therefore in what follows when writing the aforementioned identities we have this stricter interpretation of $\Theta$ in mind. We are following this convention here to avoid writing constants over and over. Of course this interpretation does not generalize to all identities involving asymptotic notations. For example consider the anomalous identity $n-n^{1-\epsilon}\ln(n)=\Theta(n)$ where the constant $\epsilon > 0$ is chosen to be sufficiently small. 
    
    For constants $k\geq 0$ and $l\geq 2$ it holds that
    \begin{align*}
    \sum_{i=1}^{\infty} \frac{\ln^k(m+i)}{(m+i)^l} &= O\Bigl(\frac{\ln^k(m)}{m^{l-1}}\Bigr).
    \end{align*}
    Similar to the two $\Theta$ identities above we have a stricter interpretation in mind when we write this identity. The interpretation is as follows. There exists a constant $\beta_3 > 0$ such that $\sum_{i=1}^{\infty} \frac{\ln^k(m+i)}{(m+i)^l} \leq  \beta_3\frac{\ln^k(m)}{m^{l-1}}$ for all $m \geq 1$.

    \begin{lemma}\label{lemma:complete expressions preliminary} For all $k \geq 2$, we can choose the constants $\alpha_2,\alpha_3,\ldots$ such that
    \begin{align*} 
    \sum_{s\in B(m,k)B(n,1)} g(s) \leq \sum_{s\in B(m,k)} g(s) + \sum_{s\in B(n,1)} g(s),
    \end{align*}
    for all $m\geq k$ and $n\geq 1$.
    \end{lemma}
    \begin{proof}
    We know the value of $\alpha_1$ from before. Fix a $k\geq 2$, and suppose that $\alpha_1,\ldots,\alpha_{k-1}$ are properly chosen. 
    First note that
    \begin{align*}
    \sum_{s\in B(m,k)} g(s) &= \sum_{i=k}^m\sum_{s\in B(i-1,k-1)} g(0^{m-i}1s).
    \end{align*}
    Also note that
    \begin{align*}
    \sum_{s\in B(m,k)B(n,1)} g(s) &= \sum_{i=k}^m\sum_{s\in B(i-1,k-1)}\sum_{t\in B(n,1)} g(0^{m-i}1st).
    \end{align*}
    Consequently we are done if we show that 
    \begin{align*}
    \sum_{s\in B(i-1,k-1)}\sum_{t\in B(n,1)} g(0^{m-i}1st) \leq \sum_{s\in B(i-1,k-1)} g(0^{m-i}1s),
    \end{align*}
    for all $i$ for which $k\leq i\leq m$. For any such $i$, the right hand side is the following
    \begin{align}
    \sum_{s\in B(i-1,k-1)} g(0^{m-i}1s) &= \sum_{j=k}^{i}\sum_{s\in B(j-2,k-2)} g(0^{m-i}1s10^{i-j})\nonumber\\
    &= \sum_{j=k}^{i}\sum_{s\in B(j-2,k-2)} \alpha_{k-1}\frac{\ln^{k-1}(j)}{j^{k-1}}\nonumber\\
    &= \sum_{j=k}^{i}\alpha_{k-1}\binom{j-2}{k-2}\frac{\ln^{k-1}(j)}{j^{k-1}}\nonumber\\
    &= \sum_{j=k}^{i}\Theta\Big(\frac{\ln^{k-1}(j)}{j}\Big)\nonumber\\
    &= \Theta\Big(\ln^{k}(i)\Big),\label{eq:complete expressions preliminary 1}
    \end{align}
    where the implied constant depends on $\alpha_{k-1}$.
    
    Fix an $i$ such that $k\leq i\leq m$, and fix a string $s\in B(i-1,k-1)$. It holds that
    \begin{align}
    \sum_{t\in B(n,1)} g(0^{m-i}1st) &= \sum_{j=1}^{n} \alpha_{k}\frac{\ln^k(i+j)}{(i+j)^k}\label{eq:complete expressions preliminary 2}\\
    &\leq \sum_{j=1}^{\infty} \alpha_{k}\frac{\ln^k(i+j)}{(i+j)^k}\nonumber\\
    &=O\Big(\frac{\ln^{k}(i)}{i^{k-1}}\Big),\nonumber 
    \end{align}
    where the implied constant depends on $\alpha_k$. Since this holds for all strings $s\in B(i-1,k-1)$ we can write
    \begin{align}
    \sum_{s\in B(i-1,k-1)}\sum_{t\in B(n,1)} g(0^{m-i}1st)
    &= \binom{i-1}{k-1}O\Big(\frac{\ln^{k}(i)}{i^{k-1}}\Big)\nonumber\\
    &= O(\ln^{k}(i)). \label{eq:complete expressions preliminary 3}
    \end{align}
    
    Comparing (\ref{eq:complete expressions preliminary 1})-(\ref{eq:complete expressions preliminary 3}) we can choose $\alpha_k$ so that the following holds
    \begin{align*}
    \sum_{s\in B(i-1,k-1)}\sum_{t\in B(n,1)} g(0^{m-i}1st) \leq \sum_{s\in B(i-1,k-1)} g(0^{m-i}1s),
    \end{align*}
    and this completes the proof.
    \end{proof}
    
    \begin{corollary}\label{lemma:complete expressions} The function $g$ is feasible in $\mathcal{P}'(n,3)$.
    \end{corollary}

    It should be a straightforward but tedious task to use the notions we built in this section to solve the following problem.
    \begin{problem}\label{problem:complete} It should be possible to extend Lemma \ref{lemma:complete expressions preliminary} so that for constants  $k,l\geq 1$ it holds that
    \begin{align*}
    \sum_{s\in B(m,k)B(n,l)} g(s) \leq \sum_{s\in B(m,k)} g(s) + \sum_{s\in B(n,l)} g(s),
    \end{align*}
    for all $n\geq k, m\geq l$.
    \end{problem}
    
    A solution to Problem \ref{problem:complete} proves that $g$ is feasible in $\mathcal{P}'(n,k)$ for all $k \geq 0$.

\section{Hardness of solving the linear programming formulations}\label{section:complexity}
    It is true that efficient algorithms exist for solving linear programs. However we know that for some languages $L\subseteq \Sigma^n$ the size of the linear program  $\mathcal{P}(L)$, i.e., the number of constraints and variables, can be exponential in $n$. Therefore given a language $L$ it might not even be tractable to write down a specification of $\mathcal{P}(L)$ let alone solving it.\footnote{An efficient algorithm for a linear program may exist even in cases where an explicit specification of the linear program cannot be obtained efficiently. For example a linear program $\mathcal{P}$ can be solved in polynomial time if there exists a separation oracle for $\mathcal{P}$. Nevertheless we show that there should not exist such an efficient algorithm (or a separation oracle) for the linear programming formulations $\mathcal{P}(L)$ and $\mathcal{P}'(L)$ in this note (under sensible complexity theoretic assumptions).} This is compatible with what we know about the computational hardness of finding optimal regular expressions and the associated problems in \cite{meyer-the-equivalence-problem-for-regular,jiang-minimal-nfa-problems-are-hard,gruber-computational-complexity-of-nfa-minimization-for-finite-and-unary-languages,gramlich-minimizing-nfas-and-regular-expressions}.

    Given a regular expression (or a DFA or an NFA) finding a regular expression of minimum length expressing the same language is PSPACE-complete as shown in Meyer and Stockmeyer \cite{meyer-the-equivalence-problem-for-regular} and in Jiang and Ravikumar \cite{jiang-minimal-nfa-problems-are-hard}. The decision version of this problem is PSPACE-complete as well, as seen in \cite{jiang-minimal-nfa-problems-are-hard}. In the decision setting a regular expression and an integer $l$ are given, and the question is whether there exists a regular expression of length $\leq l$ expressing the same language.\footnote{This is in fact explicitly mentioned only for NFAs in \cite{jiang-minimal-nfa-problems-are-hard}.} This problem is known to remain hard even when restricting to finite languages, as seen from Gruber and Holzer \cite{gruber-computational-complexity-of-nfa-minimization-for-finite-and-unary-languages}. Clearly for a finite language $L$ one can answer the decision problem if one finds $\opt(\mathcal{P}_S(L))$. Therefore solving $\mathcal{P}_S(L)$ is as hard as the aforementioned decision problem.

\section{Future directions}\label{section:future}
    We mentioned a number of open problems in this note. The most important among these open problems is whether $\opt(\mathcal{P}'(n,k))=|R_{n,k}|=\opt(\mathcal{D}'(n,k))$. 
    
    It is interesting to know whether the method of this note could be applied to other interesting languages. One interesting language is the parity language $\{s\in \{0,1\}^n:|s|_1\text{ is even}\}$. Optimal lower bounds are known for the parity language from Ellul et al. \cite{ellul-regular-expressions} and Gruber and Johannsen \cite{Gruber2008} using techniques from circuit complexity. 
    
    It is also interesting to know whether the linear programming method can be used to study measures of descriptional complexity associated with automata.
    
    For an infinite language $L$ the closures $\mathcal{C}_0(L)$ and $\mathcal{C}(L)$ are infinite sets and they should be extended to consider the Kleene star as well. The extended optimization problem $\mathcal{P}(L)$ will be 
    \begin{align*}
        \mathcal{P}(L)\quad\text{ maximize:}\quad &\sum_{s\in L}x_s\\
        \text{subject to:}\quad & \sum_{s\in K_1K_2}x_s \leq \sum_{s\in K_1}x_s + \sum_{s\in K_2}x_s  \quad\text{ for all } (K_1,K_2)\in \mathcal{C}_c(L),\\
        &\sum_{s\in K^*} x_s \leq \sum_{s\in K} x_s\quad\text{ for all } K\in \mathcal{C}(L),\\
        &0 \leq x_s \leq |s|\quad\text{ for all } s \in \mathcal{C}_0(L).
    \end{align*}
    which contains an infinite number of variables and constraints and the sums can contain infinite number of terms. We do not know if there are infinite languages for which this formulation can be used to prove nontrivial lower bounds. 
    
    We finish with two additional open problems.
        
    \begin{problem} We mentioned a procedure for converting regular expressions to integer feasible solutions of $\mathcal{P}(L)$ and $\mathcal{P}_S(L)$? Does there exist a converse procedure? Could we come up with an algorithm that converts integer feasible solutions $(\overline{W},\overline{Y},\overline{Z})$ of $\mathcal{D}_S(L)$ to regular expressions $R$ of $L$ such that $|R|$ equals the objective value of $(\overline{W},\overline{Y},\overline{Z})$? 
    \end{problem}
        
    \begin{problem} Does there exist a constant $\alpha \geq 1$ and a (rounding) algorithm that converts (not necessarily integral) feasible solutions $(\overline{Y},\overline{Z},\overline{W})$ of $\mathcal{D}_S(L)$ to regular expressions $R$ of $L$ such that $|R|$ is $\alpha$ times the objective value of $(\overline{Y},\overline{Z},\overline{W})$?
    \end{problem}
    
    If such a procedure exists it can be used to obtain approximately optimal regular expressions. Any such procedure however must be inefficient due to the results on the hardness of approximating optimal regular expressions in Gramlich and Schnitger \cite{gramlich-minimizing-nfas-and-regular-expressions}.
        
\medskip
 
\bibliographystyle{unsrt}
\bibliography{ref}

\end{document}